\documentclass[a4paper,USenglish,cleveref, autoref, thm-restate, numberwithinsect]{lipics-v2021}

\hideLIPIcs  

\title{Search-Space Reduction Via Essential Vertices Revisited: Vertex Multicut and Cograph Deletion}

\titlerunning{Search-Space Reduction Via Essential Vertices Revisited}

\author{Bart M.\,P. Jansen}{Eindhoven University of Technology, The Netherlands}{b.m.p.jansen@tue.nl}{https://orcid.org/0000-0001-8204-1268}{\flag[3cm]{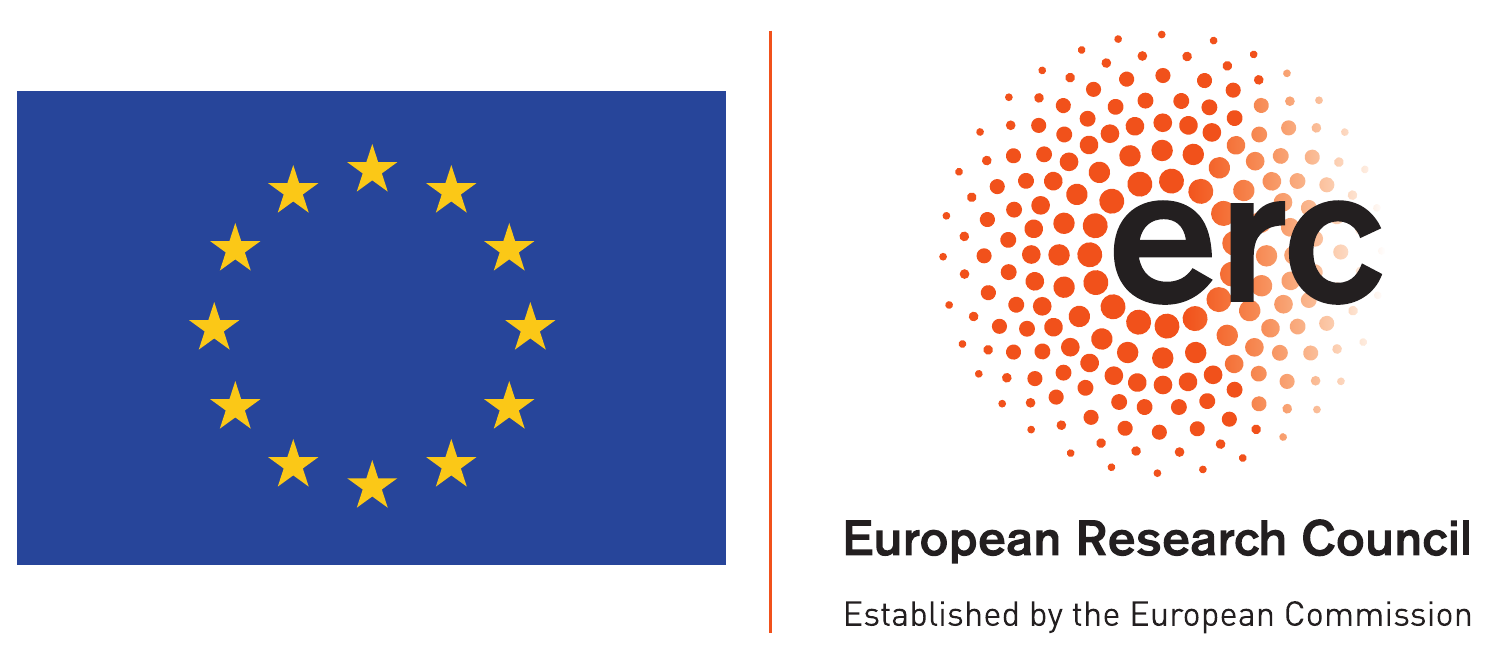}Supported by the European Research Council (ERC) under the European Union's Horizon 2020 research and innovation programme (grant agreement No 803421, ReduceSearch).} 

\author{Ruben F.\,A. Verhaegh}{Eindhoven University of Technology, The Netherlands}{r.f.a.verhaegh@tue.nl}{https://orcid.org/0009-0008-8568-104X}{}

\authorrunning{B.M.P.~Jansen and R.F.A.~Verhaegh}

\Copyright{Bart M.P.~Jansen and Ruben F.A.~Verhaegh}



\category{} 

\relatedversion{} 

\nolinenumbers 

\EventEditors{John Q. Open and Joan R. Access}
\EventNoEds{2}
\EventLongTitle{42nd Conference on Very Important Topics (CVIT 2016)}
\EventShortTitle{CVIT 2016}
\EventAcronym{CVIT}
\EventYear{2016}
\EventDate{December 24--27, 2016}
\EventLocation{Little Whinging, United Kingdom}
\EventLogo{}
\SeriesVolume{42}
\ArticleNo{23}

\newcommand{\deftask}[3]{
	\vspace{1mm}
	\noindent\fbox{
		\begin{minipage}{0.96\textwidth}
			\begin{tabular*}{\textwidth}{@{\extracolsep{\fill}}lr} \textsc{#1} &  \\ \end{tabular*}
			{\bf{Input:}} #2 \\
			{\bf{Task:}} #3
		\end{minipage}
	}
	\vspace{1mm}
}

\newcommand{\problem}[3]{
  \vspace{1mm}
  \noindent\fbox{
    \begin{minipage}{0.96\textwidth}
      \begin{tabularx}{\textwidth}{@{\hspace{\parindent}}l X}
        \multicolumn{2}{@{\hspace{\parindent}}l}{\textsc{#1}} \\
        \textbf{Input:} & #2 \\
        \textbf{Task:} & #3 \\
      \end{tabularx}
    \end{minipage}
  }
  \vspace{1mm}
}

\newcommand{\optimization}[4]{
  \vspace{1mm}
  \noindent\fbox{
    \begin{minipage}{0.96\textwidth}
      \begin{tabularx}{\textwidth}{@{\hspace{\parindent}}l X}
        \multicolumn{2}{@{\hspace{\parindent}}l}{\textsc{#1}} \\
        \textbf{Input:} & #2 \\
        \textbf{Feasible solution:} & #3 \\
        \textbf{Objective:} & #4 \\
      \end{tabularx}
    \end{minipage}
  }
  \vspace{1mm}
}

\newcommand{\C}{\mathcal{C}}
\newcommand{\R}{\mathbb{R}}
\newcommand{\Oh}{\mathcal{O}}
\renewcommand{\P}{\mathbb{P}}
\newcommand{\eps}{\varepsilon}
\newcommand{\psdetectshort}[1]{\textsc{$c$-Essential detection}}
\newcommand{\psdetect}[1]{\textsc{$#1$-Essential detection for $\Pi$}}
\newcommand{\psdetectFor}[2]{\textsc{$#1$-Essential detection for #2}}

\begin{document}

\maketitle

\begin{abstract}
For an optimization problem~$\Pi$ on graphs whose solutions are vertex sets, a vertex~$v$ is called \emph{$c$-essential} for~$\Pi$ if all solutions of size at most~$c \cdot \textsc{opt}$ contain~$v$. Recent work showed that polynomial-time algorithms to detect $c$-essential vertices can be used to reduce the search space of fixed-parameter tractable algorithms solving such problems parameterized by the size~$k$ of the solution. We provide several new upper- and lower bounds for detecting essential vertices. For example, we give a polynomial-time algorithm for \textsc{$3$-Essential detection for Vertex Multicut}, which translates into an algorithm that finds a minimum multicut of an undirected $n$-vertex graph~$G$ in time~$2^{\Oh(\ell^3)} \cdot n^{\Oh(1)}$, where~$\ell$ is the number of vertices in an optimal solution that are \emph{not} $3$-essential. Our positive results are obtained by analyzing the integrality gaps of certain linear programs. Our lower bounds show that for sufficiently small values of~$c$, the detection task becomes NP-hard assuming the \emph{Unique Games Conjecture}. For example, we show that \textsc{($2-\eps$)-Essential detection for Directed Feedback Vertex Set} is NP-hard under this conjecture, thereby proving that the existing algorithm that detects $2$-essential vertices is best-possible.
\end{abstract}

\section{Introduction}
\label{sec:introduction}

Preprocessing is an important tool for dealing with NP-hard problems. The idea is that before starting a time-consuming computation on an input, one first exhaustively applies simple transformation steps that provably do not affect the desired output, but which make the subsequently applied solver more efficient. Preprocessing is often highly effective in practice~\cite{AchterbergBGRW16,Weihe98}.

There have been several attempts to theoretically explain the speed-ups obtained by preprocessing. The concept of kernelization~\cite{Fellows06,FominLSM19}, phrased in the language of parameterized complexity theory~\cite{CyganFKLMPPS15,DowneyF12}, is one such attempt. Recently, Bumpus, Jansen, and de Kroon~\cite{bumpus-essential} proposed an alternative framework for developing and analyzing polynomial-time preprocessing algorithms that reduce the search space of subsequently applied algorithms for NP-hard graph problems. They presented the first positive and negative results in this framework, which revolves around the notion of so-called \emph{$c$-essential vertices}. In this paper, we revisit this notion by providing new preprocessing results and new hardness proofs.

To be able to discuss our results, we first introduce and motivate the concept of $c$-essential vertices and the corresponding algorithmic preprocessing task. Our results apply to optimization problems on graphs in which the goal is to find a minimum-size vertex set that hits all obstacles of a certain kind. The \textsc{(Undirected) Vertex Multicut} problem is a prime example. Given an undirected graph~$G$, annotated by a collection~$\mathcal{T}$ consisting of pairs of terminal vertices, the goal is to find a minimum-size vertex set whose removal disconnects all terminal pairs. The decision version of this problem is NP-complete, but \emph{fixed-parameter tractable} parameterized by the size of the solution: there is an algorithm by Marx and Razgon~\cite{MarxR14} that, given an $n$-vertex instance together with an integer~$k$, runs in time~$2^{\Oh(k^3)} \cdot n^{\Oh(1)}$ and outputs a solution of size at most~$k$, if one exists. The running time therefore scales exponentially in the size of the solution, but polynomially in the size of the graph. This yields a great potential for preprocessing: if an efficient preprocessing phase manages to identify some vertices~$S \subseteq V(G)$ that are guaranteed to be part of an optimal solution, then finding a solution of size~$k$ in~$G$ reduces to finding a solution of size~$k - |S|$ in~$G - S$, thereby reducing the running time of the applied algorithm and its search space. To be able to give guarantees on the amount of search-space reduction achieved, the question becomes: under which conditions can a polynomial-time preprocessing algorithm identify vertices that belong to an optimal solution?

\subparagraph*{Essential vertices} The approach that Bumpus et al.~\cite{bumpus-essential} take when answering this question originates from the idea that it may be feasible to detect vertices as belonging to an optimal solution when they are \emph{essential} for making an optimal solution. This is formalized as follows. For a real number~$c \geq 1$ and fixed optimization problem~$\Pi$ on graphs whose solutions are vertex subsets, a vertex~$v$ of an input instance~$G$ is called \emph{$c$-essential} if vertex~$v$ is contained in all $c$-approximate solutions to the instance. Hence a $c$-essential vertex is not only contained in all optimal solutions, but even in all solutions whose size is at most~$c \cdot \textsc{opt}$. To obtain efficient preprocessing algorithms with performance guarantees, the goal then becomes to develop polynomial-time algorithms to detect $c$-essential vertices in the input graph, when they are present.

For some problems like \textsc{Vertex Cover} (in which the goal is to find a minimum-size vertex set that intersects each edge), it is indeed possible to give a polynomial-time algorithm that, given a graph~$G$, outputs a set~$S$ of vertices that is part of an optimal vertex cover and contains all 2-essential vertices. For optimization problems whose structure is more intricate, like \textsc{Odd Cycle Transversal}, finding $c$-essential vertices from scratch still seems like a difficult task. Bumpus et al.~\cite{bumpus-essential} therefore formulated a slightly easier algorithmic task related to detecting essential vertices and proved that solving this simpler  task is sufficient to be able to achieve search-space reduction. For a vertex hitting set problem~$\Pi$ whose input is a (potentially annotated) graph~$G$ and whose solutions are vertex sets hitting all (implicitly defined) constraints, we denote by~$\textsc{opt}_\Pi(G)$ the cardinality of an optimal solution to~$G$. The detection task is formally defined as follows, for each real~$c \geq 1$.

\deftask
{\psdetect{c}}
{A (potentially annotated) graph $G$ and integer~$k$.}
{Find a vertex set~$S \subseteq V(G)$ such that:
    \begin{enumerate}[G1]
        \item if~$\textsc{opt}_\Pi(G) \leq k$, then there is an optimal solution in~$G$ containing all of~$S$, and \label{prop:c-essential-sup-det-double-guarantee-1} \label{g1}
        \item if~$\textsc{opt}_\Pi(G) = k$, then~$S$ contains all $c$-essential vertices.\label{prop:c-essential-sup-det-double-guarantee-2} \label{g2}
    \end{enumerate}}

The definition above simplifies the detection task by supplying an integer~$k$ in addition to the input graph, while only requiring the algorithm to work correctly for certain ranges of~$k$. The intuition is as follows: when~$k$ is correctly guessed as the size of an optimal solution, the preprocessing algorithm should find all $c$-essential vertices, and is allowed to find additional vertices as long as they are part of an optimal solution.  Bumpus et al.~\cite{bumpus-essential} give a \emph{dove-tailing}-like scheme  that manages to use algorithms for \psdetect{c} to give improved fixed-parameter tractable running times for solving~$\Pi$ from scratch. The exponential dependence of the running time of the resulting algorithm is not on the \emph{total} size of the solution, but only on the number of vertices in the solution that are \emph{not} $c$-essential. Hence their results show that large optimal solutions can be found efficiently, as long as they are composed primarily out of $c$-essential vertices. For example, they prove that a minimum vertex set intersecting all odd cycles (a solution to \textsc{Odd Cycle Transversal}) can be computed in time $2.3146^\ell \cdot n^{\Oh(1)}$, where~$\ell$ is the number of vertices in an optimal solution that are not $2$-essential and which are therefore avoided by at least one 2-approximation. Apart from polynomial-time algorithms for \psdetect{c} for various combinations of~$\Pi$ and~$c$, they also prove several \emph{lower bounds}. One of their main lower bounds concerns the \textsc{Perfect Deletion} problem, whose goal is to obtain a perfect graph by vertex deletions~\cite{HeggernesHJKV13}. They rule out the existence of a polynomial-time algorithm for \psdetectFor{c}{Perfect Deletion} for any~$c \geq 1$, assuming FPT~$\neq$~W[1]. (They even rule out detection algorithms running in time~$f(k) \cdot n^{\Oh(1)}$ for some function~$f$.)

We continue exploring the framework of search-space reduction by detecting essential vertices, from two directions. We provide both upper bounds (new algorithms for \psdetect{c}) as well as lower bounds. We start by discussing the upper bounds.

\subparagraph*{Our results: Upper bounds} The \textsc{Vertex Multicut} problem is the subject of our first results. The problem played a pivotal role in the development of the toolkit of parameterized algorithms for graph separation problems and stood as a famous open problem for years, until being independently resolved by two teams of researchers~\cite{BousquetDT18,MarxR14}. The problem is not only difficult to solve exactly, but also to approximate: Chawla et al.~\cite{ChawlaKKRS06} proved that, assuming Khot's~\cite{ugc-intro} \emph{Unique Games Conjecture} (UGC), it is NP-hard to approximate the edge-deletion version of the problem within any constant factor. A simple transformation shows that the same holds for the vertex-deletion problem. 

Our first result (\cref{thm:undirected-multicut-essential-detection}) is a polynomial-time algorithm for \textsc{$3$-Essential detection for Vertex Multicut}, which is obtained by analyzing the integrality gap of a restricted type of linear program associated with the problem. Using known results, this preprocessing algorithm translates directly into search-space reduction for the current-best FPT algorithms for solving \textsc{Vertex Multicut}. This results in an algorithm (\cref{cor:multicut:searchspace}) that computes an optimal vertex multicut in an $n$-vertex graph in time~$2^{\Oh(\ell^3)} \cdot n^{\Oh(1)}$, where~$\ell$ is the number of vertices in an optimal solution that are \emph{not} 3-essential.

Our approach for essential detection also applies for the variation of \textsc{Vertex Multicut} on \emph{directed} graphs. Since the directed setting is more difficult to deal with, vertices have to be slightly more essential to be able to detect them, resulting in a polynomial-time algorithm for \textsc{$5$-Essential detection for Directed Vertex Multicut} (\cref{thm:dvm-detection-positive}). This detection algorithm does not directly translate into running-time guarantees for FPT algorithms, though, as \textsc{Directed Vertex Multicut} is W[1]-hard parameterized by the size of the solution~\cite{PilipczukW18a}. (When the solution is forbidden from deleting terminals, the directed problem is already $W[1]$-hard with four terminal pairs, although the case of three terminal pairs is FPT~\cite{HatzelJLMPSS23}.)

Our second positive result concerns the \textsc{Cograph (Vertex) Deletion} problem. Given an undirected graph~$G$, it asks to find a minimum-size vertex set~$S$ such that~$G-S$ is a cograph, i.e., the graph~$G - S$ does not contain the 4-vertex path~$P_4$ as an induced subgraph. The problem is motivated by the fact that efficient algorithms for solving optimization problems on cographs can often be extended to work on graphs which are \emph{close} to being cographs, as long as a deletion set is known~\cite[\S 6]{Cai03a}. The decision version of \textsc{Cograph Deletion} is NP-complete due to the generic results of Lewis and Yannakakis~\cite{c-deletion}. Parameterized by the size~$k$ of the desired solution, \textsc{Cograph Deletion} is fixed-parameter tractable via the method of bounded-depth search trees~\cite{Cai96}: branching on vertices of a~$P_4$ results in a running time of~$4^k \cdot n^{\Oh(1)}$. Nastos and Gao~\cite{NastosG12} proposed a refined branching strategy by exploiting the structure of \emph{$P_4$-sparse graphs}, improving the running time to $3.115^k \cdot n^{\Oh(1)}$, following earlier improvements via the interpretation of \textsc{Cograph Deletion} as a \textsc{$4$-Hitting Set} problem~\cite{Fernau10,GrammGHN04,NiedermeierR03}. The latter viewpoint also gives a simple polynomial-time $4$-approximation. Whether a $(4-\eps)$-approximation can be computed in polynomial time is unknown; Drescher poses this \cite[\S8 Question 5]{Drescher} as an open problem for \emph{vertex-weighted} graphs. 

Our second result (\cref{lem:cograph-deletion-integrality-gap}) is a polynomial-time algorithm for \textsc{$3.5$-\textsc{Essential detection for Cograph Deletion}}. It directly translates into an FPT algorithm (\cref{cor:cograph:deletion:searchspace}) that, given a graph~$G$, outputs a minimum set~$S$ for which~$G-S$ is a cograph in time~$3.115^\ell \cdot n^{\Oh(1)}$; here~$\ell$ is the number of vertices in an optimal solution that are \emph{not} $3.5$-essential. Similarly as for \textsc{Vertex Multicut}, our detection algorithm arises from a new bound of $2.5$ on the integrality gap of a restricted version of a natural linear-programming relaxation associated to the deletion problem.

The fact that our algorithm detects $3.5$-essential vertices is noteworthy. It is known~\cite[\S 8]{bumpus-essential} that for any~$c \geq 1$, an algorithm for \psdetect{c} follows from an algorithm that computes a factor-$c$ approximation for the problem of finding a minimum-size solution avoiding a given vertex~$v$. In this setting, a $4$-approximation algorithm for \textsc{Cograph Deletion} easily follows since the problem is a special case of \textsc{$d$-Hitting Set}. We consider it interesting that we can obtain a detection algorithm whose detection constant~$c = 3.5$ is strictly better than the best-known approximation ratio~$4$ for the problem. 

Since our positive results all arise from bounding the integrality gap of certain restricted LP-formulations, we also study the integrality gap of a standard \textsc{Cograph Deletion}~LP and prove it to be~$4$ (\cref{thm:standard:cograph:gap}) using the probabilistic method. This provides a sharp contrast to the gap of~$2.5$ in our restricted setting.

\subparagraph*{Our results: Lower bounds} 

Our second set of results concerns lower bounds, showing that for certain combinations of $\Pi$ and~$c$ there are no efficient algorithms for \psdetect{c} under common complexity-theoretic hypotheses. In their work, Bumpus et al.~\cite{bumpus-essential} identified several problems~$\Pi$ such as \textsc{Perfect Deletion} for which the detection problem is intractable for \emph{all} choices of~$c$. Their proofs are based on the hardness of FPT-approximation for \textsc{Dominating Set}~\cite{KarthikLM19}. The setting for our lower bounds is different. We analyze problems for which the detection task is polynomial-time solvable for \emph{some} essentiality threshold~$c$, and investigate whether polynomial-time algorithms can exist for a smaller threshold~$c' < c$.

Our most prominent lower bound concerns the \textsc{Directed Feedback Vertex Set} problem (DFVS), which has attracted a lot of attention from the parameterized complexity community~\cite{ChenLLOR08,LokshtanovMRSZ21}. It asks for a minimum vertex set~$S$ of a \emph{directed} graph~$G$ for which~$G-S$ is acyclic. Svensson proved that under the UGC~\cite{ugc-dfvs-gap}, the problem is NP-hard to approximate to within any constant factor. Nevertheless, a polynomial-time algorithm for \psdetectFor{2}{DFVS} was given by Bumpus et al.~\cite[Lemma 3.3]{bumpus-essential}. We prove (\cref{thm:dfvs-detection-hard}) that the detection threshold~$2$ achieved by their algorithm is likely optimal: assuming the UGC, the detection problem for~$c' = 2 - \eps$ is NP-hard for any~$\eps \in (0, 1]$. To prove this, we show that an algorithm with~$c' = (2-\eps)$ would be able to distinguish instances with small solutions from instances with large solutions, while the hardness of approximation result cited above~\cite{ugc-dfvs-gap} show this task to be NP-hard under the UGC.

Apart from \textsc{Directed Feedback Vertex Set}, we provide two further lower bounds. For the \textsc{Vertex Cover} (VC) problem, an algorithm to detect $2$-essential vertices is known~\cite{bumpus-essential}. Assuming the UGC, we prove (\cref{thm:vc-detection-hard}) that \psdetectFor{(1.5 - \eps)}{VC} is NP-hard for all~$\eps \in (0, 0.5]$. A simple transformation then shows \textsc{$(1.5-\eps)$-Detection for Vertex Multicut} is also NP-hard under the UGC. These bounds leave a gap with respect to the thresholds of the current-best detection algorithms (2 and 3, respectively). We leave it to future work to close the gap.

\subparagraph*{Organization} The remainder of the paper is organized as follows. In \cref{sec:preliminaries} we give preliminaries on graphs and linear programming. \cref{sec:essential-hitting-set} introduces our formalization for hitting set problems on graphs and provides the connection between integrality gaps and detection algorithms. \cref{sec:positive-results} contains our positive results, followed by the negative results in \cref{sec:hardness-results}. We conclude with some open problems in \cref{sec:conclusion}.


    
\section{Preliminaries}
\label{sec:preliminaries}
We consider finite simple graphs, some of which are directed. Directed graphs or objects defined on directed graphs will always be explicitly indicated as such. We use standard notation for graphs and parameterized algorithms. We re-iterate the most relevant terminology and notation, but anything not defined here may be found in the textbook by Cygan et~al.~\cite{CyganFKLMPPS15} or in the previous work on essential vertices~\cite{bumpus-essential}.

\subparagraph*{Graph notation} We let $P_\ell$ denote the path graph on $\ell$ vertices. The weight of a path in a vertex-weighted graph is the sum of the weights of the vertices on that path, including the endpoints. Given two disjoint vertex sets $S_1$ and $S_2$ in a \textit{(directed)} graph $G$, we call a third vertex set $X \subseteq V(G)$ a \textit{(directed)} $(S_1, S_2)$-separator in $G$ if it intersects every \textit{(directed)} $(S_1, S_2)$-path in $G$. Note that $X$ may intersect $S_1$ and $S_2$. If $S_1$ or $S_2$ is a singleton set, we may write the single element of the set instead to obtain a $(v, S_2)$-separator for example. Menger's theorem relates the maximum number of pairwise vertex-disjoint paths between two (sets of) vertices to the minimum size of a separator between those two (sets of) vertices. We consider the following formulation of the theorem:
\begin{theorem}[{\cite[Corollary 9.1a]{schrijver-combopt}}] \label{thm:menger}
    Let $G$ be a directed graph and let $s,t \in V(G)$ be non-adjacent. Then the maximum number of internally vertex-disjoint directed $(s,t)$-paths is equal to the minimum size of a directed $(s, t)$-separator that does not include $s$ or $t$.
\end{theorem}
A \textit{fractional} \textit{(directed)} $(S_1, S_2)$-separator is a weight function that assigns every vertex in a graph a non-negative weight such that every \textit{(directed)} $(S_1, S_2)$-path has a weight of at least~$1$. The total weight of a fractional \textit{(directed)} separator is the sum of all vertex weights.

\subparagraph*{Linear programming notation} We employ well-known concepts from linear programming and refer to a textbook for additional background~\cite{schrijver-combopt}. A solution to a linear program (LP) where all variables are assigned an integral value is called an \emph{integral} solution. As we only consider LPs with a one-to-one correspondence between its variables and the vertices in a graph, integral solutions admit an alternative interpretation as vertex sets: the set of vertices whose corresponding variables are assigned a positive value. We use the interpretations of integral solutions as variable assignments or vertex sets interchangeably. We say that a minimization LP has an integrality gap of at most $c$ for some $c \in \R$ if the cost of an optimal integral solution is at most $c$ times the cost of an optimal fractional solution.

\section{Essential vertices for Vertex Hitting Set problems} 
\label{sec:essential-hitting-set}

Our positive contributions all build upon the same result from Bumpus et al.~\cite[Theorem~4.1]{bumpus-essential}, which relates integrality gaps of certain LPs to the existence of $c$-\textsc{Essential detection} algorithms. A slightly generalized formulation of this can be found below as \cref{thm:v-avoiding-lp}. First, we introduce the required background and notation.

The result indicates a strategy towards constructing a polynomial-time algorithm for $c$-\textsc{Essential detection for $\Pi$} for a vertex selection problem $\Pi$, by considering a specific special variant of that problem, that we refer to as its \textit{$v$-\textsc{Avoiding}} variant. It is defined almost identically to the original problem $\Pi$, but the input additionally contains a distinguished vertex $v \in V(G)$ which is explicitly forbidden to be part of a solution. 

The original theorem from Bumpus et al.~\cite{bumpus-essential} is specifically targeted at $\C$-\textsc{Deletion} problems for \textit{hereditary} graph classes $\C$. A graph class $\C$ is said to be hereditary when it exhibits the property that all induced subgraphs of a graph in $\C$ are again in $\C$. The corresponding $\C$-\textsc{Deletion} problem is that of finding a minimum size vertex set whose removal turns the input graph into one contained in $\C$. We remark however that the theorem holds for a broader collection of problems, namely those that can be described as \textit{\textsc{Vertex Hitting Set} problems}. To define which problems qualify as a \textsc{Vertex Hitting Set} problem, we first recall the definition of the well-known optimization problem \textsc{Hitting Set}, on which our definition of \textsc{Vertex Hitting Set} problems is based.

\optimization{Hitting Set}
{A universe $U$ and a collection $\mathcal{S} \subseteq 2^U$ of subsets of $U$.}
{A set $X \subseteq U$ such that $X \cap S \neq \emptyset$ for all $S \in \mathcal{S}$.}
{Find a feasible solution of minimum size.}

We define \textit{\textsc{Vertex Hitting Set} problems} as vertex selection problems that can be described as a special case of \textsc{Hitting Set} where the universe $U$ is the vertex set of the input graph and the collection $\mathcal{S}$ is encoded implicitly by the graph. 

This definition in particular contains all $\C$-\textsc{Deletion} problems for hereditary graph classes $\C$. This is because every hereditary graph class can be characterized by a (possibly infinite) set of forbidden induced subgraphs. A graph $G$ is in $\C$ if and only if none of its induced subgraphs are isomorphic to a forbidden induced subgraph. Therefore, a $\C$-\textsc{Deletion} instance $G$ is equivalent to the \textsc{Hitting Set} instance $(V(G), \mathcal{S})$, with $\mathcal{S}$ being the collection of all the vertex subsets that induce a forbidden subgraph in $G$.

Now, as mentioned, the $v$-\textsc{Avoiding} variants of vertex selection problems are of particular interest. A useful consequence of considering \textsc{Vertex Hitting Set} problems as special cases of \textsc{Hitting Set}, is that this yields a well-defined canonical LP formulation for such problems that can easily be modified to describe their $v$-\textsc{Avoiding} variant. This LP formulation is based on the following standard LP for a \textsc{Hitting Set} instance $(U, \mathcal{S})$, which uses variables~$x_u$ for every~$u \in U$:
\begin{align*}
    &&&&&& \text{minimize}   && \sum_{u \in U} x_u &&& &&&&&\\
    &&&&&& \text{subject to:} && \sum_{u \in S} x_u \geq 1 &&& \text{for every $S \in \mathcal{S}$} &&&&& \\
    &&&&&&                   && 0 \leq x_u \leq 1 &&& \text{for every $u \in U$} &&&&&
\end{align*}

\newcommand{\LP}[3]{\ensuremath{\mathrm{LP}_#1(#2, #3)}}
To describe the $v$-\textsc{Avoiding} variant of a \textsc{Vertex Hitting Set} problem, this LP can simply be modified by adding the constraint $x_v = 0$. For a given \textsc{Vertex Hitting Set} problem $\Pi$, a graph $G$ and a vertex $v \in V(G)$, we denote the resulting LP as $\LP{\Pi}{G}{v}$.

Although the original theorem from Bumpus et al.~\cite{bumpus-essential} makes a statement about $\C$-\textsc{deletion} problems only, it is not too hard to see that this statement also holds for any other \textsc{Vertex Hitting Set} problem. We therefore present this result as the following slight generalization.

\begin{theorem} \label{thm:v-avoiding-lp}
    Let $\Pi$ be a \textsc{Vertex Hitting Set} problem and let $c \in \R_{\geq 1}$. Then there exists a polynomial-time algorithm for $(c+1)$-\textsc{Essential detection for $\Pi$} if the following two conditions are met:
    \begin{enumerate}
        \item For all $G$ and $v \in V(G)$, there is a polynomial-time separation oracle for $\LP{\Pi}{G}{v}$.
        \item For all $G$ and $v \in V(G)$ for which $\{v\}$ solves $\Pi$ on $G$, the integrality gap of $\LP{\Pi}{G}{v}$ is at most $c$.
    \end{enumerate}
\end{theorem}

This statement admits a proof that is almost identical to the proof by Bumpus et al.~\cite[Theorem~4.1]{bumpus-essential}. At any point in that proof where the assumption is used that $\Pi$ is a $\C$-\textsc{Deletion} problem for some hereditary $\C$, this assumption may be replaced by the property that any superset of a solution to $\Pi$ is also a solution. This property is satisfied for every \textsc{Vertex Hitting Set} problem. Otherwise, no changes to the proof are required. We therefore refer the reader to this prior work for the details of the proof.

Many known results about the approximation of \textsc{Hitting Set} or about the integrality gap of \textsc{Hitting Set} LPs consider the restriction to $d$-\textsc{Hitting Set}. This is the problem obtained by requiring every $S \in \mathcal{S}$ in the input to be of size at most $d$ for some positive integer $d$. Both upper bounds and lower bounds are known for the integrality gaps of the standard LP describing $d$-\textsc{Hitting Set} instances. The standard LP is the linear program given above for the general \textsc{Hitting Set} problem.

It is well-known that this LP has an integrality gap of at most $d$ and that there exist instances for which this bound is tight. This result is for example mentioned as an exercise in a book on approximation algorithms~\cite[Exercise 15.3]{vazirani-approx}, framed from the equivalent perspective of the \textsc{Set Cover} problem.

\section{Positive results}
\label{sec:positive-results}

This section contains our positive results for essential vertex detection. For three different problems $\Pi$ and corresponding values of $c$, we provide polynomial-time algorithms for $c$-\textsc{Essential detection for $\Pi$}. The first two of these, being strongly related, are presented in \cref{ssec:multicut-positive}. There, we provide $c$-\textsc{Essential detection} algorithms for \textsc{Vertex Multicut} and \textsc{Directed Vertex Multicut} with $c=3$ and $c=5$ respectively. Afterward, we provide a $3.5$-\textsc{Essential detection} algorithm for the \textsc{Cograph Deletion} problem in \cref{ssec:cograph-deletion}.

\subsection{Vertex Multicut}
\label{ssec:multicut-positive}

Our first two positive results concern the well-studied \textsc{Vertex Multicut} problem and its directed counterpart \textsc{Directed Vertex Multicut}. These are optimization problems defined as follows.

\problem{\textit{(Directed)} Vertex Multicut}
{A \textit{(directed)} graph $G$ and a set of \textit{(ordered)} vertex pairs $\mathcal{T} = \{(s_1, t_1), \ldots, (s_r, t_r)\}$ called the terminal pairs.}
{Find a minimum size vertex set $S \subseteq V(G)$ such that there is no $(s_i, t_i) \in \mathcal{T}$ for which $G - S$ contains a \textit{(directed)} $(s_i, t_i)$-path.}

We start by observing that both problems are \textsc{Vertex Hitting Set} problems: if we let~$\mathcal{P}_\mathcal{T}(G)$ be the collection of vertex subsets that form a \textit{(directed)} $(s_i, t_i)$-path in $G$, then the \textsc{\textit{(Directed)} Vertex Multicut} instance $(G, \mathcal{T})$ is equivalent to the \textsc{Hitting Set} instance~$(V(G), \mathcal{P}_\mathcal{T}(G))$. This interpretation of the problems as special cases of \textsc{Hitting Set} is also captured by the standard LP formulations of the problems, on which the $v$-\textsc{Avoiding} LP below is based:
\begin{align*}
    &&&&&& \text{minimize}   && \sum_{u \in V(G)} x_u &&& &&&&&\\
    &&&&&& \text{subject to:} && \sum_{u \in V(P)} x_u \geq 1 &&& \text{for every \textit{(directed)} path $P$ from some $s_i$ to $t_i$}&&&&& \\
    &&&&&&                   && x_v = 0 \\
    &&&&&&                   && 0 \leq x_u \leq 1 &&& \text{for } u \in V(G) 
\end{align*}
The set of constraints in this LP formulation not only depends on the structure of the input graph $G$, but also on the set $\mathcal{T}$ of terminal pairs. Hence, we denote the LP above as $\mathrm{LP_{VM}}(G, \mathcal{T}, v)$ for undirected $G$ or as $\mathrm{LP_{DVM}}(G, \mathcal{T}, v)$ for directed $G$. The standard LP formulations of \textsc{Vertex Multicut} and \textsc{Directed Vertex Multicut} are obtained by simply removing the constraint $x_v = 0$.

\subparagraph*{The undirected case} We start with the undirected version of the problem and show in \cref{lem:undirected-multicut-integrality-gap} that $\mathrm{LP_{VM}}(G, \mathcal{T}, v)$ has an integrality gap of at most $2$ for all \textsc{Vertex Multicut} instances $(G, \mathcal{T})$ where $v \in V(G)$ is such that $\{v\}$ is a solution. This bound yields a polynomial-time algorithm for $3$-\textsc{Essential detection for Vertex Multicut} as presented in \cref{thm:undirected-multicut-essential-detection}.

\begin{lemma} \label{lem:undirected-multicut-integrality-gap}
    Let $(G, \mathcal{T})$ be a \textsc{Vertex Multicut} instance with some $v \in V(G)$ such that $\{v\}$ is a solution for this instance. Then $\mathrm{LP_{VM}}(G, \mathcal{T}, v)$ has an integrality gap of at most $2$.
\end{lemma}
\begin{proof}
    Let $\mathbf{x} = (x_u)_{u \in V(G)}$ be an optimal solution to $\mathrm{LP_{VM}}(G, \mathcal{T}, v)$ and let $z = \sum_{u \in V(G)} x_u$ be its value. If we interpret the values of $x_u$, as given by $\mathbf{x}$, as vertex weights, then by definition of the LP, all $(s_i, t_i)$-paths have weight at least 1 for all $\{s_i, t_i\} \in \mathcal{T}$. Moreover, all such paths must pass through $v$ because $\{v\}$ is a solution, so we know for every $\{s_i, t_i\} \in \mathcal{T}$ that all $(s_i, v)$-paths or all $(t_i, v)$-paths (or both) have weight at least $\frac12$.

    We proceed by stating a reformulation of this property. Let $D \subseteq V(G)$ be the set of all vertices $u$ such that every $(u, v)$-path has weight at least $\frac12$. Then, the above property can also be described as follows: for every $\{s_i, t_i\} \in \mathcal{T}$, at least one of $s_i$ and $t_i$ is in $D$.

    Using this alternate formulation, it follows that every $(v, D)$-separator $X$ is also a valid solution to the given \textsc{Vertex Multicut} instance. To see this, consider an arbitrary $(s_i, t_i)$-path $P$ for some arbitrary $(s_i, t_i) \in \mathcal{T}$. Since $\{v\}$ is a solution, $P$ intersects $v$. If $s_i \in D$, then the fact that $X$ is a $(v, D)$-separator implies that $X$ intersects the subpath of $P$ between $v$ and $s_i$. The same holds for $t_i$. Since at least one of $s_i$ and $t_i$ is in $D$, it follows that $X$ must intersect $P$. Because $P$ was an arbitrary $(s_i, t_i)$-path for an arbitrary terminal pair $(s_i, t_i)$, $X$ hits all such paths and therefore it is a vertex multicut.
    
    Now to prove that $\mathrm{LP_{VM}}(G, \mathcal{T}, v)$ has an integrality gap of at most $2$, it suffices to show that there exists a $(v, D)$-separator $X \subseteq V(G)$ of size at most $2z$ that does not contain $v$. To see that this is indeed the case, we start by constructing a fractional $(v, D)$-separator $f \colon V(G) \rightarrow \R$ of weight at most $2z$ and with $f(v) = 0$. We obtain $f$ by simply doubling the values given by $\mathbf{x}$, i.e.: $f(u) := 2 x_u$ for all $u \in V(G)$. This step is inspired by a proof from Golovin, Nagarajan, and Singh that shows an upper bound on the integrality gap of a \textsc{Multicut} variant in trees \cite{doubling-trick}.
    
    We observe that indeed $f(v) = 2 \cdot x_v = 0$, since $\mathbf{x}$ is a solution to $\mathrm{LP_{VM}}(G, \mathcal{T}, v)$, which requires that $x_v = 0$. Furthermore, $D$ was constructed such that all paths from $v$ to a vertex in $D$ have a weight of at least $\frac12$ under the vertex weights as given by $\mathbf{x}$. Hence, under the doubled weights of $f$, all such paths have a weight of at least $1$, witnessing that $f$ is in fact a fractional $(v, D)$-separator.

    The final step of the proof is now to show that the existence of this \textit{fractional} $(v, D)$-separator of weight $2z$ implies the existence of an \textit{integral} $(v, D)$-separator of size at most $2z$ that does not contain~$v$. To do so, we use Menger's theorem on the auxiliary directed graph~$G'$ obtained from $G$ by turning all undirected edges into bidirected edges, while adding a sink node $t$ with incoming edges from all vertices in $D$.

    Consider a maximum collection $\mathcal{P}$ of internally vertex-disjoint directed $(v, t)$-paths in~$G'$. Let~$X \subseteq V(G') \setminus \{v,t\}$ be a directed~$(v,t)$-separator in~$G'$ of size~$|\mathcal{P}|$, whose existence is guaranteed by~\cref{thm:menger}. The construction of~$G'$ ensures that~$X$ is a~$(v,D)$-separator in~$G$ that does not contain~$v$, and therefore corresponds to an integral solution to $\mathrm{LP_{VM}}(G, \mathcal{T}, v)$. To bound the integrality gap by~$2$, it therefore suffices to prove that~$|\mathcal{P}| = |X| \leq 2z$.

    For each $(v,t)$-path~$P \in \mathcal{P}$ in~$G'$, the prefix obtained by omitting its endpoint~$t$ yields a~$(v,D)$-path in~$G$. Since~$f$ is a fractional $(v, D)$-separator, it must assign every such prefix of~$P \in \mathcal{P}$ a weight of at least~$1$. Because $f(v) = 0$ and because the paths in~$\mathcal{P}$ are internally vertex-disjoint, we find that the total weight of $f$ must be at least~$|\mathcal{P}| = |X|$. Since the weight of~$f$ is at most~$2z$, we find that $|\mathcal{P}| = |X| \leq 2z$. This concludes the proof.
\end{proof}

We can even construct \textsc{Vertex Multicut} instances $(G, \mathcal{T})$ that are solved by some $\{v\} \subseteq V(G)$ for which the integrality gap of $\mathrm{LP_{VM}}(G, \mathcal{T}, v)$ is arbitrarily close to $2$, showing that the bound in \cref{lem:undirected-multicut-integrality-gap} is tight. To construct such an instance, let $G$ be a (large) star graph, let $v \in V(G)$ be its center and let $\mathcal{T} = \binom{V(G) \setminus \{v\}}{2}$. Clearly, $\{v\}$ is a solution to the \textsc{Vertex Multicut} instance $(G, \mathcal{T})$. 

To determine the integrality gap of $\mathrm{LP_{VM}}(G, \mathcal{T}, v)$, we first note that any solution to the \textsc{Vertex Multicut} instance that avoids $v$ must, at least, include all but one of the leaves from $G$. Any such set is indeed a solution, which shows that the smallest integral solution to $\mathrm{LP_{VM}}(G, \mathcal{T}, v)$ has value $|V(G)| - 2$. A smaller fractional solution to the program may be obtained by assigning every leaf of $G$ a value of $\frac12$, which would yield a solution with a total value of $\frac12 \cdot (|V(G)| - 1)$. Observe that such a construction of $G$, $\mathcal{T}$, and $v$ can be used to get an LP with an integrality gap arbitrarily close to $2$ by having the star graph $G$ be arbitrarily large.

Regardless of the bound on the integrality gap being tight, \cref{lem:undirected-multicut-integrality-gap} and \cref{thm:v-avoiding-lp} combine to prove the following result.

\begin{theorem}\label{thm:undirected-multicut-essential-detection}
    There exists a polynomial-time algorithm for $3$-\textsc{Essential detection for Vertex Multicut}.
\end{theorem}
\begin{proof}
    As shown in \cref{lem:undirected-multicut-integrality-gap}, $\mathrm{LP_{VM}}(G, \mathcal{T}, v)$ has an integrality gap of at most $2$ when~$\{v\}$ is a solution to the \textsc{Vertex Multicut} instance $(G, \mathcal{T})$. Secondly, this LP has a very simple polynomial time separation oracle in general: for every terminal pair $\{s_i, t_i\}$, we can find a minimum weight $(s_i, t_i)$-path in polynomial time using a shortest path algorithm and check whether it has at least weight $1$. Although most shortest-path algorithms are defined for edge-weighted graphs, it takes only minor modifications to make standard algorithms, like Dijkstra's algorithm \cite{dijkstras-algorithm}, work on vertex-weighted graphs as well. Hence, it follows from \cref{thm:v-avoiding-lp} that there is a polynomial-time algorithm for $3$-\textsc{Essential detection for Vertex Multicut}.
\end{proof}

The algorithm to detect $3$-essential vertices leads in a black-box fashion to a search-space reduction guarantee for the current-best algorithm for solving \textsc{Vertex Multicut} due to Marx and Razgon~\cite{MarxR14}. This follows from a result of Bumpus et al.~\cite[Theorem 5.1]{bumpus-essential} (\cref{sec:essential:to:search:space}). While they originally stated their connection between essential detection and search-space reduction for \textsc{$\mathcal{C}$-Deletion} problems, it is easy to see that the same proof applies for any \textsc{Vertex Hitting Set} problem: the only property of \textsc{$\mathcal{C}$-Deletion} that is used in their proof is that for any vertex set~$X \subseteq V(G)$, a vertex set~$Y \subseteq V(G-X)$ is a solution to~$G-X$ if and only if~$X \cup Y$ is a solution to~$G$; this property holds for any \textsc{Vertex Hitting Set} problem.

\begin{corollary} \label{cor:multicut:searchspace}
    There is an algorithm that, given a \textsc{Vertex Multicut} instance~$(G, \mathcal{T})$ on~$n$ vertices, outputs an optimal solution in time~$2^{\Oh(\ell^3)} \cdot n^{\Oh(1)}$, where~$\ell$ is the number of vertices in an optimal solution that are \emph{not} 3-essential.
\end{corollary}

\subparagraph*{The directed case} Keeping in mind the techniques used to prove \cref{lem:undirected-multicut-integrality-gap}, we proceed to the next problem: \textsc{Directed Vertex Multicut}. By similar arguments, we find the $v$-\textsc{Avoiding} LP of this problem to have a bounded integrality gap as well. However, adaptations to these arguments are required to take the directions of edges into consideration, yielding a higher bound on the integrality gap of the directed version of the problem.

\begin{lemma} \label{lem:directed-multicut-integrality-gap}
    Let $(G, \mathcal{T})$ be a \textsc{Directed Vertex Multicut} instance with some $v \in V(G)$ such that $\{v\}$ is a solution for it. Then $\mathrm{LP_{DVM}}(G, \mathcal{T}, v)$ has an integrality gap of at most $4$.
\end{lemma}
\begin{proof}
    Our proof starts similarly to the proof of \cref{lem:undirected-multicut-integrality-gap}. We let $\mathbf{x} = (x_u)_{u \in V(G)}$ be an optimal solution to the $v$-\textsc{Avoiding} LP and let $z = \sum_{u \in V(G)} x_u$ be its value. Again, we interpret the values given to each of the vertices by $\mathbf{x}$ as vertex weights. Then by definition of the LP, all directed $(s_i, t_i)$-paths have weight at least $1$ for all $(s_i, t_i) \in \mathcal{T}$. Moreover, because $\{v\}$ is a solution, all such paths must pass through~$v$. Hence, we know that for every $(s_i, t_i) \in \mathcal{T}$ all directed $(s_i, v)$-paths or all directed $(t_i, v)$-paths have weight at least $\frac12$.
     
    We again proceed by stating a reformulation of this property. This time, however, we define two sets of vertices, to cope with the directionality of paths in the graph. We let $S \subseteq V(G)$ be the set of all vertices $u$ such that every directed $(u, v)$-path has weight at least $\frac12$ and we let $T$ be the set of all vertices $u$ such that every directed $(v, u)$-path has weight at least $\frac12$. Then, the above property can also be described as follows: for every $(s_i, t_i) \in \mathcal{T}$, we have that $s_i \in S$ or $t_i \in T$.

    This allows us to see that every set $X$ that is both a directed $(S, v)$-separator and a directed $(v, T)$-separator is also a valid solution for the given \textsc{Directed Vertex Multicut} instance. Consider to this end an arbitrary directed path $P$ from $s_i$ to $t_i$ for some arbitrary $(s_i, t_i) \in \mathcal{T}$. Since $\{v\}$ is a solution, $P$ intersects $v$. If $s_i \in S$, then the fact that $X$ is a directed $(S, v)$-separator implies that $X$ intersects the prefix of $P$ from $s_i$ to $v$. Likewise, if $t_i \in T$, then the fact that $T$ is a directed $(v,T)$-separator implies that $X$ intersects the suffix of $P$ from $v$ to $t_i$. Since at least one of the two holds, $X$ intersects $P$. Because $P$ was an arbitrary directed $(s_i, t_i)$-path for an arbitrary terminal pair $(s_i, t_i)$, $X$ hits all such paths and therefore it is a vertex multicut.
     
    Now to prove that $\mathrm{LP_{DVM}}(G, \mathcal{T}, v)$ has an integrality gap of at most $4$, it suffices to show that there exists such a set $X \subseteq V(G)$ of size at most $4z$ that does not contain $v$. To see that this is indeed the case, we start by defining an assignment $f \colon V(G) \rightarrow \R$ of fractional values to the vertices, of total weight $2z$, that describes both a fractional directed $(S, v)$-separator and a fractional directed $(v, T)$-separator, which has $f(v) = 0$. We obtain $f$ by doubling the values given by $\mathbf{x}$, i.e.: $f(u) := 2x_u$ for all $u \in V(G)$. 

    Of course, $\sum_{u \in V(G)} f(u) = 2 \cdot \sum_{u \in V(G)} x_u = 2z$. We also observe that $f(v) = 2 \cdot x_v = 0$, since $\mathbf{x}$ is a solution to $\mathrm{LP_{DVM}}(G, \mathcal{T}, v)$, which requires that $x_v = 0$. Furthermore, $S$ was defined such that all directed paths from $S$ to $v$ have a weight of at least $\frac12$ under the vertex weights as given by $\mathbf{x}$. Hence, under the doubled weights of~$f$, all such paths have a weight of at least $1$, witnessing that $f$ is a fractional directed $(S, v)$-separator. Likewise, all paths from $v$ to a vertex in $T$ have a weight of at least $1$ under the doubled weights of $f$, witnessing that $f$ is also a fractional directed $(v, T)$-separator.

    Next, we show that the existence of the fractional separator $f$ implies the existence of an integral separator that is at most~$4$ times as large. We do this by showing that a directed $(S, v)$-separator $X_S$ and a directed $(v, T)$-separator $X_T$ exist that are each of size at most $2z$ and that do not include $v$. These two separators combine to form a set that is both a directed $(S, v)$-separator and a directed $(v, T)$-separator. As this results in a separator of size at most $4z$ that does not include $v$, this suffices to prove the statement.
     
    We start by proving that $X_S$ exists as described and we do so by applying Menger's theorem on an auxiliary graph obtained from $G$. This graph $G'$ may be obtained from $G$ by adding a source node $s$ to the graph with outgoing arcs to all vertices in~$S$. Consider a maximum collection~$\mathcal{P}$ of internally vertex-disjoint directed $(s, v)$-paths in~$G'$. Let~$X_S \subseteq V(G') \setminus \{s, v\}$ be a directed $(s, v)$-separator in~$G'$ of size~$|\mathcal{P}|$, whose existence is guaranteed by \cref{thm:menger}. The construction of~$G'$ ensures that~$X_S$ is a directed $(S, v)$-separator in $G$ that does not contain $v$. We proceed by showing that $|X_S| = |\mathcal{P}| \leq 2z$.

    For each directed $(s, t)$-path $P \in \mathcal{P}$ in $G'$, the suffix obtained by omitting its first vertex $s$ yields a directed $(S, v)$-path in $G$. Since $f$ is a fractional directed $(S, v)$-separator, it must assign every such suffix of a path $P \in \mathcal{P}$ a weight of at least $1$. Because $f(v) = 0$ and because the paths in $\mathcal{P}$ are internally vertex-disjoint, we find that the total weight of $f$ must be at least $|\mathcal{P}| = |X|$. Since the weight of $f$ is at most $2z$, we find that $|X| = |\mathcal{P}| \leq 2z$.

    By similar arguments, we can show that there also exists a directed $(v, T)$-separator $X_T$ in $G$ of size at most $2z$ that does not include $v$. To do so, we would construct $G'$ from $G$ not by adding the source node $s$, but by adding a sink node $t$ with incoming arcs from all vertices in $T$. The remaining argumentation follows analogously.

    Having shown the existence of the separators $X_S$ and $X_T$, we can combine them to conclude that $X_S \cup X_T$ is both a directed $(S, v)$-separator and a directed $(v, T)$-separator that does not include $v$ and that $|X_S \cup X_T| \leq |X_S| + |X_T| \leq 4z$. This concludes the proof.
 \end{proof}

Similar to the undirected setting, this upper bound on the integrality gap leads to the following algorithmic result when combined with \cref{thm:v-avoiding-lp}.

\begin{theorem} \label{thm:dvm-detection-positive}
    There exists a polynomial-time algorithm for $5$-\textsc{Essential detection for Directed Vertex Multicut}.
\end{theorem}
This statement admits a proof that is almost identical to the proof of \cref{thm:undirected-multicut-essential-detection}, since the shortest-path algorithm that provides the separation oracle of the \textsc{Vertex Multicut} LP can also take directed graphs as input.

\subsection{Cograph Deletion}
\label{ssec:cograph-deletion}

Our next positive result concerns the \textsc{Cograph Deletion} problem. As this is a specific case of $\mathcal{C}$-\textsc{Deletion}, this is more in line with the original research direction for $c$-\textsc{Essential detection} introduced by Bumpus et al.~\cite{bumpus-essential}, where a framework was built around $\mathcal{C}$-\textsc{Deletion} problems. The \textsc{Cograph Deletion} problem is defined as follows.

\problem{Cograph Deletion}
{An undirected graph $G$.}
{Find a minimum size set $S \subseteq V(G)$ such that $G - S$ is a cograph (i.e.: $G-S$ does not contain a path on $4$ vertices as an induced subgraph).}

We start by observing that the \textsc{Cograph Deletion} problem is a \textsc{Vertex Hitting Set} problem: if we let $\mathcal{P}_4(G)$ be the collection of vertex subsets that induce a $P_4$ in $G$, then the \textsc{Cograph Deletion} instance $G$ is equivalent to the \textsc{Hitting Set} instance $(V(G), \mathcal{P}_4(G))$. Again, motivated by \cref{thm:v-avoiding-lp}, we study the $v$-\textsc{Avoiding} LP for this problem:
\begin{align*}
    &&&&&& \text{minimize}   && \sum_{u \in V(G)} x_u &&& &&&&&\\
    &&&&&& \text{subject to:} && \sum_{u \in V(H)} x_u \geq 1 &&& \text{for every induced subgraph $H$ of $G$ isomorphic to $P_4$}&&&&& \\
    &&&&&&                   && x_v = 0 \\
    &&&&&&                   && 0 \leq x_u \leq 1 &&& \text{for } u \in V(G) &&&&&
\end{align*}
For a given graph $G$ and vertex $v \in V(G)$, we denote the LP above as $\mathrm{LP_{CD}}(G, v)$. Whenever $v$ is such that $G - v$ is a cograph, the resulting LP admits a simple upper bound on the integrality gap. This bound is derived from the observation that the $v$-\textsc{Avoiding Cograph Deletion} problem is a special case of $3$-\textsc{Hitting Set}: the vertex sets to be hit in the problem are the triplets of vertices that, together with $v$, induce a $P_4$ in $G$. As the natural LP describing $3$-\textsc{Hitting Set} has an integrality gap of at most~$3$, it follows that the natural LP formulation of $v$-\textsc{Avoiding Cograph Deletion}, to which the above LP is equivalent, also has an integrality gap of at most~$3$.

This section is dedicated to proving a stronger result than this trivial bound. We prove that, whenever $v$ is such that $G-v$ is a cograph, $\mathrm{LP_{CD}}(G, v)$ has an integrality gap of at most $2.5$. To prove this, we use a method inspired by iterative rounding \cite{iterative-rounding}, where an approximate integral solution can be obtained by solving the LP, picking all vertices that receive a large enough value, updating the LP to no longer contain these vertices and repeating these steps until a solution is found. 

For our purposes, we consider values of at least $0.4$ to be `large enough'. However, we will see that an extension to the original method is required since $\mathrm{LP_{CD}}(G, v)$ is not guaranteed to always have an optimal solution that assigns at least one vertex a value of~$\geq 0.4$. 
This issue is reflected in the inductive proof below by having the step case split into two subcases. The first of these deals with the standard iterative rounding setup, while the second subcase deals with the possibility of an optimal solution not assigning any vertex a large value.

\newcommand{\LPorig}{\mathrm{LP_{CD}}(G, v)}
\newcommand{\xorig}{\mathbf{x}}
\newcommand{\zorig}{z}
\newcommand{\Vlarge}{V_{\geq0.4}}
\newcommand{\LPres}{\mathrm{LP_{CD}}(G - \Vlarge, v)}
\newcommand{\zres}{z_\mathrm{res}}
\newcommand{\Vres}{V_{\mathrm{res}}}
\newcommand{\opt}[1]{\mathrm{OPT}(#1)}

\begin{lemma} \label{lem:cograph-deletion-integrality-gap}
    Let $G$ be a graph and let $v \in V(G)$ be such that $G - v$ is a cograph. Then $\mathrm{LP_{CD}}(G, v)$ has an integrality gap of at most $2.5$.
\end{lemma}
\begin{proof}
    We prove the statement by induction on the value of an optimal fractional solution to the linear program.

    First, consider as base case that $\LPorig$ has an optimal fractional solution of value $0$. Then this solution is the all-zero solution. This is also an integral optimum solution to the program, so the integrality gap of the program is $1$ and the claim holds.
    
    Next, let $\xorig = (x_u)_{u \in V(G)}$ be an optimal solution to $\LPorig$, let $\zorig = \sum_{u \in V(G)} x_u$ be its value and let $\Vlarge \subseteq V(G)$ be the set of vertices that are assigned a value of at least $0.4$ in this solution. We distinguish two cases.
    
    \subparagraph*{Case 1:} Suppose $\Vlarge \neq \emptyset$. Consider the pair $(G -  \Vlarge, v)$ and note that ${(G-\Vlarge) - v}$, being a subgraph of $G-v$, is a cograph. Also note that the restriction of~$\xorig$ to~$G - \Vlarge$ is a feasible solution to $\LPres$. This solution has a value of $\zorig - \sum_{u \in \Vlarge} x_u \leq \zorig - 0.4 \cdot |\Vlarge|$, which is strictly smaller than~$\zorig$ by the assumption that $\Vlarge \neq \emptyset$. If we let $\zres$ be the value of an optimal solution to $\LPres$, then this implies that $\zres \leq \zorig - 0.4 |\Vlarge| < \zorig$ as well.

    Then, by the induction hypothesis, an integral solution $\Vres$ to $\LPres$ with $|\Vres| \leq 2.5 \zres$ exists. To prove that $\LPorig$ has an integrality gap of at most~$2.5$, we proceed by showing that $\Vres \cup \Vlarge$ is an integral solution to $\LPorig$ with value at most $2.5 \zorig$. We start by arguing that $\Vres \cup \Vlarge$ is a valid integral solution.

    First note that neither $\Vres$ nor $\Vlarge$ contains $v$ since both $\LPres$ and $\LPorig$ require $x_v = 0$. Therefore, the union of these two sets also does not contain~$v$. Secondly, note that $\Vres$ (by construction of $\LPres$) contains a vertex from every induced~$P_4$ in~$G$ that does not already have a vertex in $\Vlarge$. As such, $\Vres \cup \Vlarge$ contains a vertex from every induced~$P_4$ in~$G$, which makes it a feasible solution to $\LPorig$. 
    
    Knowing this, it remains to prove that $\Vres \cup \Vlarge$ has size at most $2.5 \zorig$. Recall that we derived $\zres \leq \zorig - 0.4 \cdot |\Vlarge|$. We can use this inequality to make the following derivation:
    \begin{align*}
        |\Vres \cup \Vlarge| &\leq |\Vres| + |\Vlarge| 
                ~\leq~ 2.5 \zres + |\Vlarge| & \text{by definition of $\Vres$} \\
                & \leq 2.5 \left( \zorig - 0.4 \cdot |\Vlarge| \right) + |\Vlarge| & \text{by the above inequality} \\
                &= 2.5 \zorig - |\Vlarge| + |\Vlarge| ~=~ 2.5 \zorig & \text{since $2.5 \cdot 0.4 = 1$}
    \end{align*}

    \subparagraph*{Case 2:} Suppose $\Vlarge = \emptyset$. Let $V^\ast \subseteq V(G) \setminus \{v\}$ be the set of vertices that are part of at least one induced $P_4$ in $G$. To prove that $\LPorig$ has an integrality gap of at most $2.5$, we show that the smaller set of $V^\ast \cap N_G(v)$ and $V^\ast \setminus N_G(v)$ is a solution to the program with size at most $2.5 \zorig$.

    We start by proving that $V^\ast \cap N_G(v)$ and $V^\ast \setminus N_G(v)$ are both feasible solutions to $\LPorig$. We do so using an observation about the structure of the graph $P_4$. Observe that this graph has the property that each vertex has at least one neighbor and at least one non-neighbor. Since $v$ is part of every induced $P_4$ in $G$ by assumption, this means that every induced $P_4$ in $G$ contains both a neighbor and a non-neighbor of $v$.

    The above observation implies that $V^\ast \cap N_G(v)$ and $V^\ast \setminus N_G(v)$ both intersect all induced subgraphs of $G$ isomorphic to $P_4$. Hence, both of these sets are feasible solutions to $\LPorig$. It remains to prove that the smaller of the two sets has a size of at most $2.5 \zorig$. 
    
    Since $V^\ast \cap N_G(v)$ and $V^\ast \setminus N_G(v)$ form a partition of $V^\ast$ into two parts, the smaller of the two will always be of size at most $|V^\ast| / 2$. Therefore, it suffices to show that $|V^\ast| / 2 \leq 2.5 \zorig$. To prove this, we start by showing that the assumption that $\Vlarge = \emptyset$ implies that $x_w \geq 0.2$ for all vertices $w \in V^\ast$.

    We prove this property by contradiction, so suppose there is some vertex $w \in V(G) \setminus \{v\}$ that is part of an induced $P_4$, but which has $x_w < 0.2$. Let $H$ be an induced subgraph of $G$ that is isomorphic to $P_4$ and with $w \in V(H)$. Because $G - v$ is a cograph, $v$ is contained in every induced $P_4$ and in particular $v \in V(H)$. By definition of $\LPorig$, we have $x_v = 0$. By the assumption that $\Vlarge = \emptyset$, the two vertices in $V(H) \setminus \{w, v\}$ have value at most $0.4$, so $\sum_{u \in V(H)} x_u < 1$, which contradicts the validity of $\mathbf{x}$.

    Knowing that $x_w \geq 0.2$ for all $w \in V^\ast$, it follows that $\zorig = \sum_{u \in V(G)} x_u \geq 0.2 |V^\ast|$. Rewriting this inequality, we obtain $|V^\ast| / 2 \leq 2.5 \zorig$.
\end{proof}

At the moment, we are not aware of any examples of pairs $(G, v)$ where $G - v$ is a cograph and for which $\mathrm{LP_{CD}}(G, v)$ has an integrality gap of $2.5$. Therefore, the bound above does not have to be tight and the integrality gap of such programs may even be as small as $2$. However, there do exist pairs $(G, v)$ where $G - v$ is a cograph and for which $\mathrm{LP_{CD}}(G, v)$ has an integrality gap arbitrarily close to $2$.

Such a pair $(G, v)$ may be obtained by constructing~$G$ as the union of~$m$ disjoint edges and adding the vertex~$v$ to it which is adjacent to exactly one endpoint of each of these~$m$ edges. Then, any integral solution to $\mathrm{LP_{CD}}(G, v)$ must include, at least, one endpoint from all but one of the original~$m$ edges. Any such set of vertices is in fact a feasible integral solution, so a smallest integral solution has size $m-1$.

An optimal fractional solution may be obtained by assigning all~$m$ neighbors of~$v$ a value of~$0.5$, which yields a total value of $m/2$. Hence, the integrality gap of $\mathrm{LP_{CD}}(G, v)$ is $\frac{m-1}{m/2} = 2 \cdot \frac{m-1}{m}$, which can be arbitrarily close to $2$ for arbitrarily large $m$.

Like earlier, the upper bound on the integrality gap shown in~\cref{lem:cograph-deletion-integrality-gap} leads to the following algorithmic result.

\begin{theorem} \label{thm:cograph-essential-detection}
    There exists a polynomial-time algorithm for $3.5$-\textsc{Essential detection for Cograph Deletion}.
\end{theorem}
\begin{proof}
    As shown in \cref{lem:cograph-deletion-integrality-gap}, $\mathrm{LP_{CD}}(G, v)$ has an integrality gap of at most $2.5$ when $\{v\}$ is a solution to the \textsc{Cograph Deletion} instance $G$. Secondly, this LP has a very simple polynomial-time separation oracle in general: for every possible quadruplet of vertices in the graph, it can be checked in polynomial time whether they induce a $P_4$ and if so, whether their values sum to at least $1$. It then follows from \cref{thm:v-avoiding-lp} that there is a polynomial-time algorithm for $3.5$-\textsc{Essential detection for Cograph Deletion}.
\end{proof}

The algorithm to detect $3.5$-essential vertices leads to a search-space reduction guarantee for the current-best parameterized algorithm for \textsc{Cograph Deletion}~\cite{NastosG12} via Theorem 5.1 by Bumpus et al.~\cite{bumpus-essential}.

\begin{corollary} \label{cor:cograph:deletion:searchspace}
    There is an algorithm that, given a \textsc{Cograph Deletion} instance~$G$ on~$n$ vertices, outputs an optimal solution in time~$3.115^\ell \cdot n^{\Oh(1)}$, where~$\ell$ is the number of vertices in an optimal solution that are \emph{not} $3.5$-essential.
\end{corollary}

\subsubsection{The standard Cograph Deletion problem} The result from \cref{lem:cograph-deletion-integrality-gap} is interesting because of its algorithmic implications, as presented in \cref{thm:cograph-essential-detection} and \cref{cor:cograph:deletion:searchspace}. Moreover, the result is even surprising because the bound on the integrality gap is lower than one may initially think. As mentioned at the start of \cref{ssec:cograph-deletion}, a trivial bound on the integrality gap of $\mathrm{LP}_{CD}(G, v)$ for instances with $G - v$ being a cograph is $3$, because $v$-\textsc{Avoiding Cograph Deletion} is a special case of $3$-\textsc{Hitting Set}. Polynomial-time solvable LPs for $3$-\textsc{Hitting Set} in general likely cannot even have an integrality gap smaller than $3$, as mentioned in the preliminaries. However, the specific setting of the $v$-\textsc{Avoiding Cograph Deletion} problem where $v$ is defined to be a singleton solution yields an LP whose integrality gap is bounded by $2.5$.

This remainder of this section is dedicated to explaining the apparent decrease in difficulty from the general $3$-\textsc{Hitting Set} to the setting of \cref{lem:cograph-deletion-integrality-gap}. In particular, we argue that this does not follow from a property that is intrinsic to the \textsc{Cograph Deletion} problem. Rather, the smaller integrality gap must therefore be attributed to the special setting we consider: the $v$-\textsc{Avoiding} variant of the problem for instances where $\{v\}$ is a solution.

To argue this, we provide a lower bound on the integrality gap of the standard LP of the standard \textsc{Cograph Deletion} problem. This LP captures the interpretation of the problem as special case of \textsc{Hitting Set} and may be obtained from the $v$-\textsc{Avoiding} LP by dropping the constraint $x_v = 0$. Because the problem is in particular a special case of $4$-\textsc{Hitting Set}, we find a trivial upper bound on the integrality gap of $4$. The lower bound we provide in this section shows that this bound of $4$ is tight, even for the special case of \textsc{Cograph Deletion}.

This result further motivates the strategy of investigating the $v$-\textsc{Avoiding} versions of problems where $\{v\}$ is a solution. In that sense, it also serves as a parallel to the duality between the positive results from Lemmas~\ref{lem:undirected-multicut-integrality-gap} and~\ref{lem:directed-multicut-integrality-gap} and the observation that neither of the standard LPs describing \textsc{Vertex Multicut} and \textsc{Directed Vertex Multicut} have constantly bounded integrality gaps.

We prove the lower bound on the standard LP of \textsc{Cograph Deletion} in two parts. First, in \cref{lem:exists-graph-with-many-P4s}, we show that for every sufficiently large $n$, there exists a graph with $n$ vertices in which no induced subgraph with $\log^2(n)$ vertices is a cograph. Then, in \cref{lem:graph-with-many-P4s-has-high-integrality-gap}, we show that the aforementioned standard \textsc{Cograph Deletion} LP attains an integrality gap that is arbitrarily close to $4$ on such graphs for arbitrarily large $n$. For the sake of notation, we restrict $n$ to be a power of $2$, to ensure that its logarithm is integer.

\begin{lemma} \label{lem:exists-graph-with-many-P4s}
    For every sufficiently large $n$ that is a power of $2$, there is a graph $G$ on $n$ vertices such that every $X \in \binom{V(G)}{\log^2(n)}$ induces a graph containing an induced $P_4$.
\end{lemma}
\begin{proof}
    We prove the statement using the probabilistic method \cite{probabilistic-method}. That is, we show that a randomly generated graph has a strictly positive probability of having the desired property. This establishes the existence of graphs with this property in a non-constructive manner. In this case, the desired property for a graph is that every possible induced subgraph on $\log^2(n)$ vertices contains an induced $P_4$.

    Let $G$ therefore be an Erdős–Rényi random graph with $n$ vertices and edge probability $\frac12$. I.e.: it is generated by starting with $n$ vertices of degree $0$ and independently connecting every pair of two vertices with probability $\frac12$. Now, if we let
    \[
        p = \P\left( \exists X \in \binom{V(G)}{\log^2(n)} \text{ such that $G[X]$ is a cograph} \right), 
    \]
    then the probability of $G$ being a graph with the desired property is $1-p$. Hence, the proof continues by showing that $p < 1$. More strongly even, we show that $p \rightarrow 0$ as $n \rightarrow \infty$. We start by observing that the following bound holds:
    \begin{align*}
        p &= \P \left( \bigcup_{X \in \binom{V(G)}{\log^2(n)}} G[X] \text{ is a cograph} \right) \\
          &\leq \sum_{X \in \binom{V(G)}{\log^2(n)}} \P(G[X] \text{ is a cograph}) & \text{(by the union bound)}
    \end{align*}
    
    To compute the expression above, we compare the number of different possible graphs that can exist on $\log^2(n)$ vertices to the number of cographs on $\log^2(n)$ vertices. We express these quantities by counting the number of different labeled graphs and cographs. A labeled graph on some number $\ell$ vertices is a graph whose vertex set is $[\ell]$ and two labeled graphs are identical precisely when their vertex sets and edge sets are equal. Since all labeled graphs are equally likely to result from the random sampling, we can rewrite the bound from above as follows:
    \begin{align*}
        p &\leq \binom{n}{\log^2(n)} \frac{\#\text{labeled cographs on $\log^2(n)$ vertices}}{\#\text{labeled graphs on $\log^2(n)$ vertices}}
    \end{align*}
    
    We bound each of the three terms in this expression to eventually arrive at an upper bound for $p$. We start by bounding the binomial coefficient as follows:
    \[
    \binom{n}{\log^2(n)} \leq n^{\log^2(n)} = 2^{\log^3(n)}.
    \]

    Next, we compute the denominator of the fraction: the number of labeled graphs on $\log^2(n)$ vertices. In such a graph, there are $\binom{\log^2(n)}{2}$ pairs of vertices. A graph is uniquely determined by specifying which of these pairs are connected by an edge, for which there are~$2^{\binom{\log^2(n)}{2}}$ options. Hence, there are $2^{\binom{\log^2(n)}{2}} = 2^{\frac12\left( \log^4(n) - \log^2(n) \right)}$ labeled graphs on $\log^2(n)$ vertices.

    Finally, we bound the numerator of the fraction: the number of labeled cographs on $\log^2(n)$ vertices. It is known that the number of labeled cographs on $\ell$ vertices is at most~$\ell! \cdot 2^{\ell-1}$~\cite{cographs-speed-original}~(see also \cite{cographs-speed-english}). Because $\ell! \leq 2^{\ell \log(\ell)}$, this upper bound may be simplified to $2^{\ell + \ell \log(\ell)}$. For $\ell = \log^2(n)$, this is $2^{\log^2(n) + \log^2(n)\log(\log^2(n))}$.

    When we combine the bounds for each of the three different terms, we find that:
    \[
    p \leq 2^{\log^3(n)} \frac{2^{\log^2(n) + \log^2(n)\log(\log^2(n))}}{2^{\frac12\left( \log^4(n) - \log^2(n) \right)}} = \frac{2^{\Oh(\log^3(n))}}{2^{\Omega(\log^4(n))}}
    \]
    which tends to $0$ as $n \rightarrow \infty$.

    Not only does this show that there exists a graph on $n$ vertices in which every induced subgraph of size $\log^2(n)$ contains an induced $P_4$, it even shows that almost all sufficiently large graphs exhibit this property.
\end{proof}

We proceed by showing that the standard \textsc{Cograph Deletion} LP attains high integrality gaps on such graphs.
\begin{lemma} \label{lem:graph-with-many-P4s-has-high-integrality-gap}
    Let $G$ be a graph on $n$ vertices such that every $X \in \binom{V(G)}{\log^2(n)}$ induces a subgraph containing an induced $P_4$. Then the integrality gap of the standard \textsc{Cograph Deletion} LP on $G$ is at least $4 \frac{n - \log^2(n)}{n}$.
\end{lemma}
\begin{proof}
    To prove a lower bound on the integrality gap of the \textsc{Cograph Deletion} LP, we give separate bounds for the optimal values of integral and fractional solutions on $G$. These are a lower bound for the value of an optimal integral solution and an upper bound for the value of an optimal fractional solution. We present the bounds in this order.

    First, we observe that every integral solution must have a value of at least $n - \log^2(n)$. After all, if there are at least $\log^2(n)$ vertices not in the solution, then by assumption, these vertices induce a subgraph that still contains an induced $P_4$.

    Secondly, we remark that setting $x_u = \frac14$ for every $u \in V(G)$ is a valid fractional solution to the LP. This assignment ensures that the sum of any four distinct variables is $1$, which suffices to satisfy the constraints of the LP. This solution has a value of $n / 4$, meaning that this value serves as an upper bound on the value of an optimal integral solution.

    Combining these two bounds, we find that the \textsc{Cograph Deletion} LP attains an integrality gap of at least $\frac{n - \log^2(n)}{n / 4} = 4 \frac{n - \log^2(n)}{n}$ on $G$.
\end{proof}

The following theorem now follows easily from the ingredients above.

\begin{theorem} \label{thm:standard:cograph:gap}
    For all $\eps > 0$, the integrality gap of the standard \textsc{Cograph Deletion} LP is larger than $4 - \eps$. 
\end{theorem}
\begin{proof}
    We proceed by contradiction, so suppose that there is some $\eps > 0$ such that this integrality gap is at most $4 - \eps$. To arrive at a contradiction, we argue that a graph exists on which the LP attains an integrality gap greater than $4 - \eps$.

    First, let $n$ be a power of $2$ such that $4 \frac{n - \log^2(n)}{n} > 4 - \eps$. Since the left-hand side of this inequality approaches $4$ as $n$ grows to infinity, such $n$ exists. From \cref{lem:exists-graph-with-many-P4s}, we know that there exists a graph $G$ on $n$ vertices with the property that every $X \in \binom{V(G)}{\log^2(n)}$ induces a subgraph containing an induced $P_4$. Then, it follows from \cref{lem:graph-with-many-P4s-has-high-integrality-gap} that the LP attains an integrality gap of at least $4 \frac{n - \log^2(n)}{n}$ on $G$. By definition of $n$, this value is strictly greater than $4 - \eps$. This is a contradiction.
\end{proof}

\section{Hardness results}
\label{sec:hardness-results}

In this section, we show two main hardness results regarding essential detection algorithms. The first of these concerns \textsc{Directed Feedback Vertex Set} (DFVS). The objective in this problem is to find a smallest vertex set $S$ in a directed input graph $G$ such that $G - S$ is acyclic. We slightly abuse notation by using the acronym DFVS to denote both a (not necessarily optimal) solution to a given input and the name of the problem itself. Additionally, we let $\mathrm{DFVS}(G)$ denote the size of a smallest DFVS in $G$. The hardness result obtained for DFVS can be extended to \textsc{Directed Vertex Multicut}. The second result concerns \textsc{Vertex Cover} (VC) and it can be extended to other vertex hitting set problems on undirected graphs, including \textsc{Vertex Multicut}.

Our results are based on the hardness assumption posed by the Unique Games Conjecture (UGC) \cite{ugc-intro}. Although the conjecture has remained open since its introduction in 2002, many conditional hardness results in the area of approximation algorithms follow from it. Before stating our first new hardness result, we mention the known result it is derived from, which itself is an implication of the UGC. By the nature of the UGC, many results derived from it show the conditional hardness of distinguishing between two types of problem inputs: one with a very small solution and one with a very large solution. Indeed, we derive our hardness from one such result due to Svensson~\cite{ugc-dfvs-gap} that implies the following.

\begin{lemma}[{\cite[Theorem 1.1]{ugc-dfvs-gap}}]  \label{lem:dfvs-ugc-original}
    Assuming the UGC, the following problem is NP-hard for any integer $r\geq2$ and sufficiently small constant $\delta>0$. Given a directed graph on $n$ vertices, distinguish between the following two cases:
    \begin{enumerate}[(a)]
        \item There exists a partition of $V(G)$ into $V_0, V_1, \ldots, V_r$ such that: \label{itm:dfvs-ugc-original-a}
        \begin{enumerate}[(i)]
            \item $|V_i| \geq \frac{1 - \delta}{r} n$ for every $i = 1, \ldots, r$ and \label{itm:dfvs-ugc-original-a1}
            \item $G - V_0 - V_i$ is acyclic for every $i = 1, \ldots, r$. \label{itm:dfvs-ugc-original-a2}
        \end{enumerate}

        \item $\mathrm{DFVS}(G) \geq (1 - \delta)n$. \label{itm:dfvs-ugc-original-b}
    \end{enumerate}
\end{lemma}

Note that case (a) implies that a small DFVS exists in $G$, as we derive below. Although the original statement is slightly stronger than that, the existence of this small DFVS is sufficient for our application. As such, we rephrase the result as follows to highlight the distinction between one case having a small DFVS, while the other case only contains large DFVSs.

\begin{lemma} \label{lem:dfvs-ugc}
    Assuming the UGC, the following problem is NP-hard for any integer $r\geq2$ and sufficiently small constant~$\delta > 0$. Given a directed $n$-vertex graph $G$, distinguish between the following two cases:
    \begin{itemize}
        \item $\mathrm{DFVS}(G) \leq \left( \frac{1-\delta}{r} + \delta \right) n$
        \item $\mathrm{DFVS}(G) \geq (1 - \delta) n$
    \end{itemize}
\end{lemma}
\begin{proof}
    To prove that this reformulation of \cref{lem:dfvs-ugc-original} also holds, we show that its case~(\ref{itm:dfvs-ugc-original-a}) implies the existence of a DFVS in $G$ of size at most $\left( \frac{1-\delta}{r} + \delta \right) n$. To this end, let $V_0, V_1, \ldots, V_r$ be a partition of $V(G)$ as described in case (\ref{itm:dfvs-ugc-original-a}). Then, for any $i \in [r]$, $V_0 \cup V_i$ is a DFVS by property~(ii) and
    \begin{align*}
    |V_0 \cup V_i|  &= n - \left| \bigcup_{j \in [r] \setminus \{i\}} V_j \right| 
                    ~\leq~ n - (r-1) \cdot \frac{1 - \delta}{r}n 
                    ~=~ n - \left( (1- \delta) - \frac{1-\delta}{r}\right) n \\
                    &= \left( 1 - 1 + \delta + \frac{1-\delta}{r} \right) n 
                    ~=~ \left( \frac{1-\delta}{r} + \delta \right) n.
    \end{align*}
    The inequality in the derivation above follows from property~(i).
\end{proof}
We use this formulation to prove the first hardness results.

\newcommand{\Qin}{Q_\mathrm{in}}
\newcommand{\Qout}{Q_\mathrm{out}}
\begin{theorem} \label{thm:dfvs-detection-hard}
    Assuming the UGC, $(2-\eps)$-\textsc{Essential detection for DFVS} is NP-hard for any $\eps \in (0, 1]$.
\end{theorem}
\begin{proof}
    Let $\eps \in (0,1]$ be given. We can assume w.l.o.g. that~$\frac{2}{\eps}$ is integral. If not, we could consider some~$\eps' < \eps$ such that~$\frac{2}{\eps'}$ is integer and prove hardness for $(2 - \eps')$-essential detection. As a $(2 - \eps)$-essential detection algorithm is also a valid algorithm for $(2 - \eps')$-essential detection, this would imply the hardness of $(2 - \eps)$-essential detection as well.
    
    Now, we use \cref{lem:dfvs-ugc} as a starting point for hardness. To do so, let $G$ be an arbitrary directed graph on $n$ vertices. To use \cref{lem:dfvs-ugc}, we show how to reduce $G$ into a directed graph $G'$, such that solving $(2 - \eps)$-\textsc{Essential detection for DFVS} on $G'$ allows us to distinguish between $\mathrm{DFVS}(G) \leq \left( \frac{1 - \delta}{r} + \delta \right)n$ and $\mathrm{DFVS}(G) \geq (1 - \delta) n$ for some integer $r \geq 2$ and arbitrarily small $\delta > 0$. We assume w.l.o.g. that $n \cdot \eps / 2$ is integer. If not, we could consider the graph obtained by having $2/\eps$ independent copies of $G$ instead, as the minimum size of a DFVS relative to the total graph size would be the same. We proceed by explaining the reduction, after which we prove its correctness.

    Our reduction starts with the directed graph~$G$ and depends on the value of $\eps$. See also \cref{fig:dfvs-reduction} for a visual example. We start the construction of~$G'$ as a copy of~$G$. To avoid confusion between vertices in~$G$ and~$G'$, we denote the current vertex set of~$G'$ as~$P$. Next, we expand the graph with two additional sets of vertices~$\Qin$ and~$\Qout$. These sets each consist of $m := (1 - \frac{\eps}{2})n$ vertices, which is integer by our assumptions on~$n$ and~$\eps$. We denote the vertices of $\Qin$ as $q_1, \ldots, q_m$ and the vertices of $\Qout$ as $q_1', \ldots, q_m'$. We define~$Q := \Qin \cup \Qout$.
    
    We complete the construction of $G'$ by adding more arcs to it. For every $i \in [m]$, we add the arc $(q_i, q_i')$. For every $p \in P$ and $q_i \in \Qin$, we add the arc $(p, q_i)$. For every $p \in P$ and $q_i' \in \Qout$, we add the arc $(q_i', p)$. This completes the construction of $G'$. Observe that it ensures that $(p, q_i, q_i')$ is a directed cycle for every $p \in P$ and $i \in [m]$.

    \begin{figure}[t]
        \centering
        \begin{subfigure}[t]{0.48\textwidth}
            \centering
            \includegraphics[width=\textwidth]{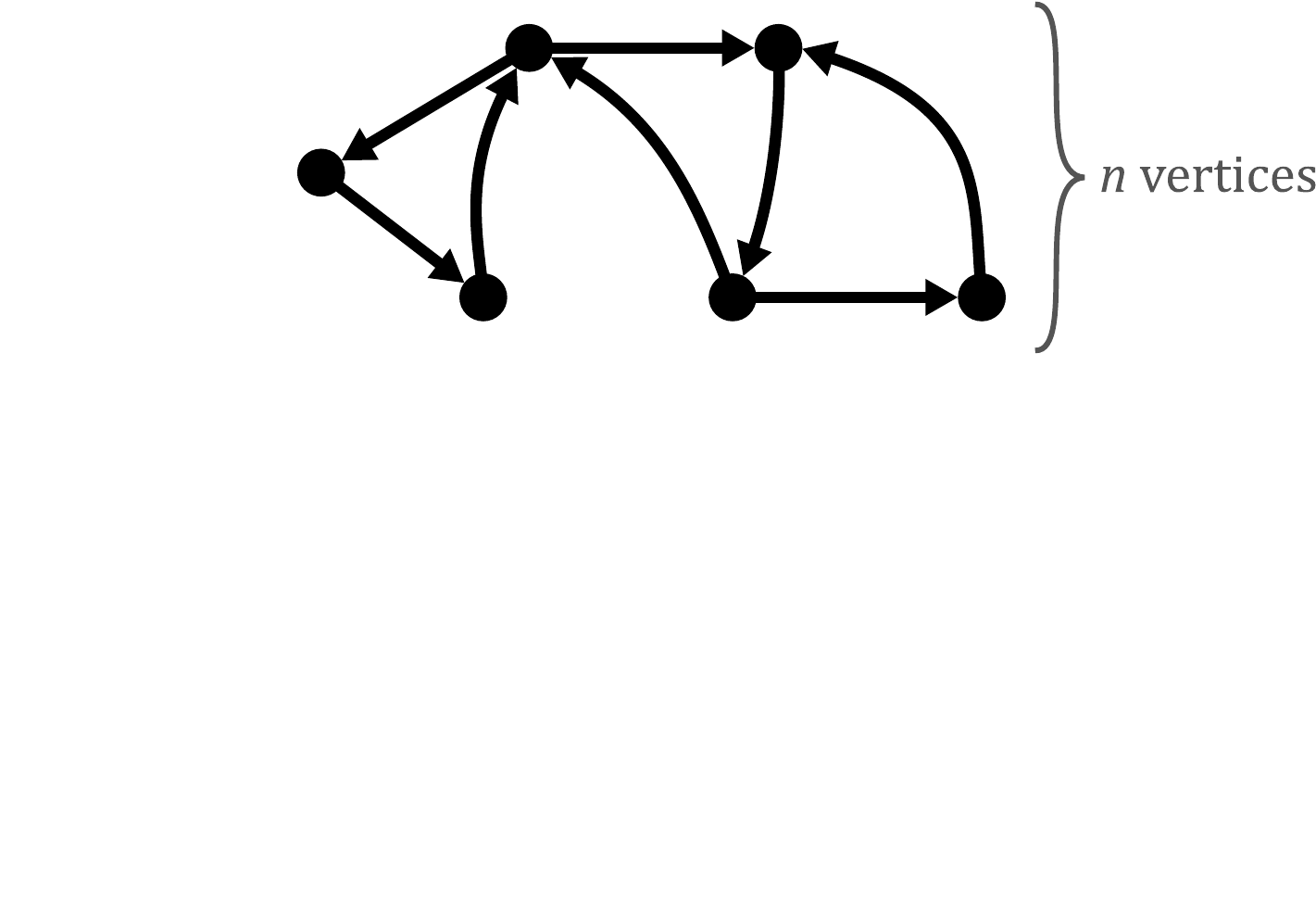}
            \caption{Example input graph $G$ with $n$ vertices.}
        \end{subfigure}
        \hfill
        \begin{subfigure}[t]{0.48\textwidth}
            \centering
            \includegraphics[width=\textwidth]{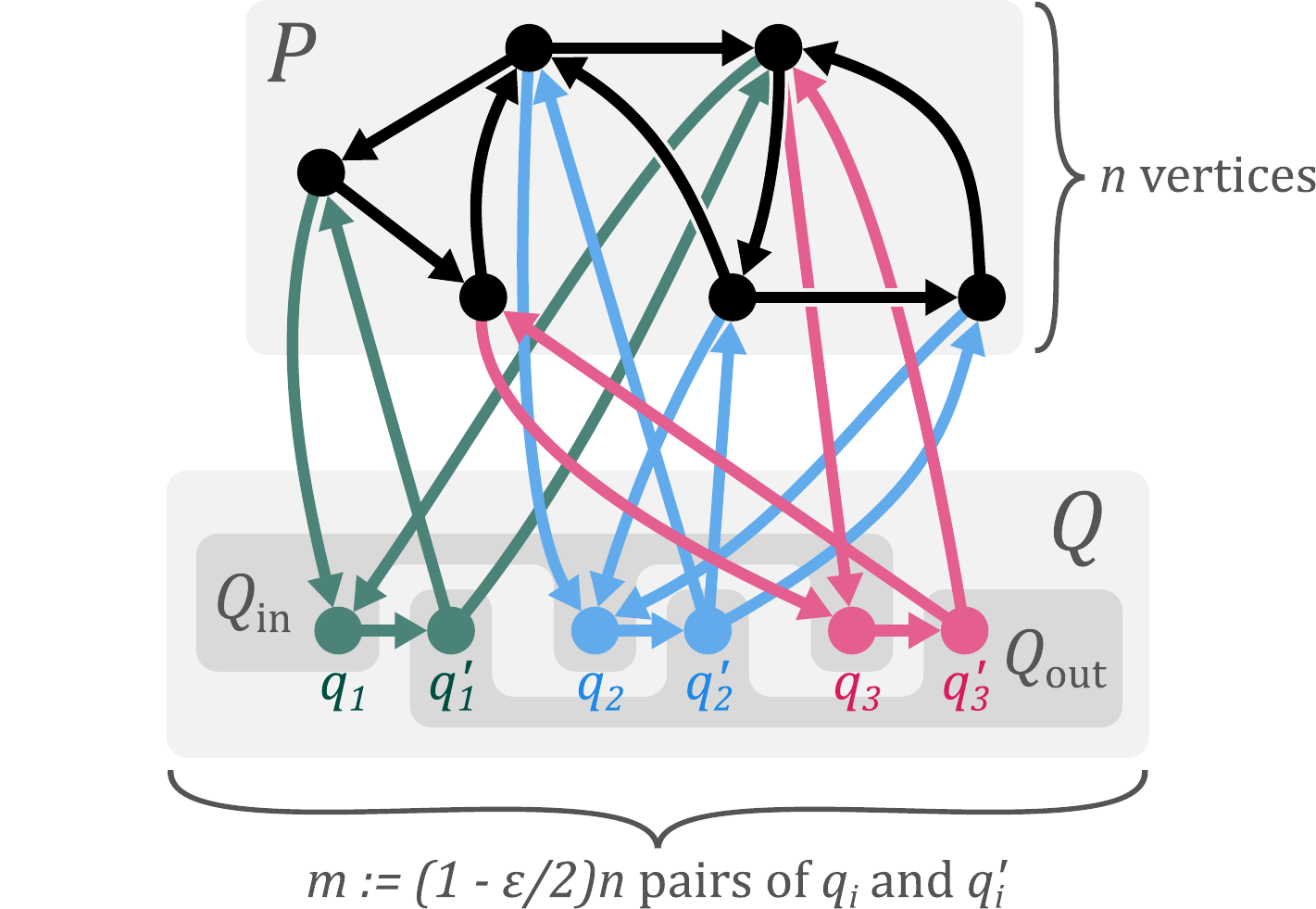}
            \caption{The resulting graph $G'$ when performing the reduction on $G$ with $\eps = 1$.}
        \end{subfigure}
        \caption{An example of the reduction for DFVS. Note: for the sake of visualization, not all arcs are shown in the drawing of $G'$. Every $q_i \in \Qin$ should have incoming arcs from all $u \in P$ and every $q_i' \in \Qout$ should have outgoing arcs to all $u \in P$. The colors serve only to improve readability: they distinguish which arcs are incident to which $(q_i, q_i')$ pair.}
        \label{fig:dfvs-reduction}
    \end{figure}

    To prove the correctness of this reduction, we show that the output of an algorithm for $(2 - \eps)$-\textsc{Essential detection for DFVS} on $G'$ can be used as subroutine to distinguish between  $\mathrm{DFVS}(G) \leq \left( \frac{1 - \delta}{r} + \delta \right)n$ and $\mathrm{DFVS}(G) \geq (1 - \delta) n$ for some integer~$r \geq 2$ and arbitrarily small~$\delta > 0$. In particular, we show that this is possible for $r = \frac{4}{\eps}$, which is integer by the assumption that $\frac{2}{\eps}$ is integer. From now on, we fix~$r = \frac{4}{\eps}$ and~$\delta > 0$ to be arbitrarily small so that $\delta \leq \frac{\eps}{4}$ in particular.

    Now, suppose that an algorithm for $(2 - \eps)$-\textsc{Essential detection for DFVS} exists and let $S \subseteq V(G')$ be its output when run on $G'$ with $k$ set to $n$. (Recall, $k$ represents a guess for (an upper bound) of the size of an optimal solution in $G'$. In this setting, that would be a guess for the size of a minimum size DFVS in $G'$.) We show that the following two claims hold.
    \begin{claim} \label{clm:dfvs-small}
        If $\mathrm{DFVS}(G) \leq \left( \frac{1 - \delta}{r} + \delta \right)n$, then $|S| < n$.
    \end{claim}
    \begin{claim} \label{clm:dfvs-large}
        If $\mathrm{DFVS}(G) \geq (1 - \delta) n$, then $|S| = n$.
    \end{claim}
    Then, simply checking the size of the output set $S$ suffices to distinguish between $\mathrm{DFVS}(G) \leq \left( \frac{1 - \delta}{r} + \delta \right)n$ and $\mathrm{DFVS}(G) \geq (1 - \delta) n$. From \cref{lem:dfvs-ugc}, we know that this distinction is NP-hard to make under the UGC, meaning that $(2 - \eps)$-\textsc{Essential detection for DFVS} is also NP-hard when assuming the UGC. To prove \cref{thm:dfvs-detection-hard}, it therefore suffices to prove \cref{clm:dfvs-small} and \cref{clm:dfvs-large}.
    \begin{claimproof}[Proof of \cref{clm:dfvs-small}]
        Suppose that $\mathrm{DFVS}(G) \leq \left( \frac{1 - \delta}{r} + \delta \right)n$. We show that $G'$ has a DFVS whose size is strictly smaller than $n$. Then, by property~(G\ref{g1}), we know that $S$ must be a subset of some smallest DFVS in $G'$. This implies that $|S| < n$ and hence proves the claim. Therefore, it suffices and remains to prove that $G'$ has a DFVS whose size is strictly smaller than $n$.

        Let $X$ be a smallest DFVS in $G$. We show that $X \cup \Qin$ is a DFVS in $G'$ and that its size is smaller than $n$. We start by showing that $X \cup \Qin$ is a DFVS in $G'$, i.e.: we show that $G' - X - \Qin$ is acyclic. We proceed by contradiction, so suppose there is some cycle $C$ in this graph. First note that $C$ does not contain any vertices from $\Qout$, since these vertices only had incoming arcs from $\Qin$ in $G'$. Hence, in the graph $G' - X - \Qin$, the vertices from $\Qout$ only have outgoing arcs, which means that they are not part of any cycles and in particular not of $C$. It follows that $C$ must live in $G'[P] - X$. But because $G'[P] = G$, $C$ is also a cycle in $G - X$, which contradicts the assumption that $X$ is a DFVS in $G$. Now, it remains to show that $|X \cup \Qin| < n$, which is achieved through the following derivation.
        \begin{align*}
            |X \cup \Qin| &= |X| + |\Qin| 
                      = \mathrm{DFVS}(G) + \left( 1 - \eps/2 \right) n \\
                      &\leq \left( \frac{1 - \delta}{r} + \delta \right)n + \left( 1 - \eps/2 \right) n \\
                      &< \left( \frac{1}{r} + \delta \right)n + (1 - \eps / 2)n & \text{since $\delta > 0$}\\
                      &\leq \left( \frac{1}{4 / \eps} + \eps/4 + 1 - \eps/2 \right) n = n & \text{since $r = 4/\eps$ and $\delta \leq \eps / 4$ \claimqedhere}
        \end{align*}
    \end{claimproof}

    \begin{claimproof}[Proof of \cref{clm:dfvs-large}]
        Suppose that $\mathrm{DFVS}(G) \geq (1 - \delta) n$. We show that the following two properties hold:
        \begin{enumerate}[(a)]
            \item $P$ is a DFVS in $G'$. \label{itm:large-dfvs-a}
            \item All vertices in $P$ are $(2 - \eps)$-essential in $G'$. \label{itm:large-dfvs-b}
        \end{enumerate}
        Before showing that these two properties hold, we argue why doing so would suffice to prove the claim. First note that $P$ is the unique smallest DFVS in $G'$ if both properties hold: $P$ is a DFVS by property~(\ref{itm:large-dfvs-a}) and any solution that does not contain every vertex of $P$ must be at least $(2 - \eps)$ times as large by property~(\ref{itm:large-dfvs-b}). Hence, $\mathrm{DFVS}(G') = |P| = n$. Then by property~(G\ref{g2}), the set~$S$ (the output of running a $(2 - \eps)$-\textsc{Essential detection for DFVS} algorithm on $G'$ with $k$ set to $n$) must contain all of $P$, so $|S| \geq n$. By property~(G\ref{g1}), no other vertices can be in $S$, so $|S| = n$.
    
        We proceed by showing properties~(\ref{itm:large-dfvs-a}) and~(\ref{itm:large-dfvs-b}) hold. It is easy to see that $P$ is a DFVS in $G'$, since $G' - P$ consists only of pairwise disjoint arcs and is therefore acyclic. This shows that property~(\ref{itm:large-dfvs-a}) holds.
    
        To show that property~(\ref{itm:large-dfvs-b}) holds, consider an arbitrary $u \in P$ and let $X'$ be a DFVS in $G'$ that does not contain $u$. We show that $|X'| > (2 - \eps) n$, meaning that $u$ is $(2 - \eps)$-essential. As $u \in P$ was chosen arbitrarily, all vertices in $P$ are $(2 - \eps)$-essential. 
        
        To prove that $|X'| > (2 - \eps) n$, we provide lower bounds for both $|X' \cap P|$ and $|X' \cap Q|$. First, we bound $|X' \cap P|$. Since $G'[P] = G$, we must have $|X' \cap P| \geq \mathrm{DFVS}(G)$, which is at least $(1 - \delta)n$ by assumption.
        
        Secondly, we prove that $|X' \cap Q| \geq m$. We proceed by contradiction, so suppose otherwise. Then, there must be at least one $i \in [m]$ such that both $q_i,q_i' \notin X'$. Since $u \notin X'$ by definition of $X'$, $(u, q_i, q_i')$ is a cycle that lives in $G' - X'$. This contradicts the assumption that $X'$ is a DFVS in $G'$. Therefore, it follows that $|X' \cap Q| \geq m = (1 - \frac{\eps}{2})n$. Now, we can make the following derivation.
        \begin{align*}
            |X'| &= |X' \cap P| + |X' \cap Q| \geq (1 - \delta)n + (1 - \eps/2)n \\
                 &\geq (1 - \eps/4 + 1 - \eps/2)n > (2 - \eps) n & \text{since $\delta \leq \eps/4$}
        \end{align*}
        This shows that property~(\ref{itm:large-dfvs-b}) holds, which, together with property~(\ref{itm:large-dfvs-a}), proves the claim.
    \end{claimproof}

    This concludes the proof of \cref{thm:dfvs-detection-hard}.
\end{proof}

The lower bound of \cref{thm:dfvs-detection-hard} yields an analogous lower bound for \textsc{Directed Vertex Multicut}, since the set of solutions for \textsc{Directed Feedback Vertex Set} on a graph~$G$ equals the set of solutions to the \textsc{Directed Vertex Multicut} instance obtained from~$G$ by introducing a terminal pair~$(u_2, u_1)$ for every arc~$(u_1, u_2)$ of~$G$.

\begin{corollary}
    Assuming the UGC, $(2 - \eps)$-\textsc{Essential detection for Directed Vertex Multicut} is NP-hard for any $\eps > 0$.
\end{corollary}

By applying the proof technique above, but starting from a result about hardness of approximation for \textsc{$d$-Hitting Set}~\cite{vc-ugc}, we prove the lower bounds for \textsc{Vertex Cover} and \textsc{Undirected Vertex Multicut}. To state the hardness result, let $\mathrm{HS}(U, \mathcal{S})$ denote the size of a smallest solution to the \textsc{Hitting Set} instance $(U, \mathcal{S})$.
\begin{lemma}[\cite{vc-ugc}] \label{lem:vc-ugc}
    Assuming the UGC, the following problem is NP-hard for any sufficiently small $\delta>0$.
    Given a $d$-\textsc{Hitting Set} instance $(U, \mathcal{S})$ with $|U| = n$, distinguish between the following two cases:
    \begin{itemize}
        \item $\mathrm{HS}(U, \mathcal{S}) \leq (\frac1d + \delta)n$
        \item $\mathrm{HS}(U, \mathcal{S}) \geq (1 - \delta)n$
    \end{itemize}
\end{lemma}
As \textsc{Vertex Cover} is the special case of $d$-\textsc{Hitting Set} where $d = 2$, we can use the statement above to prove the result below. In its proof, we use $\textsc{VC}(G)$ to denote the minimum size of a vertex cover in $G$. 

\begin{theorem} \label{thm:vc-detection-hard}
    Assuming the UGC, $(1.5 - \eps)$-\textsc{Essential detection for VC} is NP-hard for any $\eps \in (0, 0.5]$.
\end{theorem}
\begin{proof}
    Let $\eps \in (0, 0.5]$ be given. We can assume w.l.o.g. that $\frac{4}{\eps}$ is integral. If not, we could consider some $\eps' < \eps$ such that $\frac{4}{\eps'}$ is integer and prove hardness for $(2 - \eps')$-essential detection. As a $(1.5 - \eps)$-essential detection algorithm is also a valid algorithm for $(1.5 - \eps')$-essential detection, this would imply the hardness of $(1.5 - \eps)$-essential detection as well.

    We use \cref{lem:vc-ugc} as a starting point for hardness. To do so, let $G$ be an arbitrary graph on $n$ vertices. To use \cref{lem:vc-ugc}, we show how to reduce $G$ into $G'$, such that solving $(1.5 - \eps)$-\textsc{Essential detection for VC} on $G'$ allows us to distinguish between $\mathrm{VC}(G) \leq \left( \frac12 + \delta \right) n$ and $\mathrm{VC}(G) \geq (1 - \delta) n$ for arbitrarily small $\delta$. We assume w.l.o.g. that $n \cdot \eps/4$ is integer. If not, we could consider the graph obtained by having $4/\eps$ independent copies of $G$ instead, as the minimum size of a vertex cover relative to the total graph size would be the same as in $G$. We proceed by explaining the reduction, after which we prove its correctness.

    Our reduction starts with the directed graph $G$ and depends on the value of $\eps$. See also \cref{fig:vc-reduction} for a visual example. We start the construction of $G'$ as a copy of $G$. To avoid confusion between vertices in $G$ and $G'$, we denote the current vertex set of $G'$ as $P$. Next, we expand the graph with a set $Q$ of $m := \left( \frac12 - \frac{\eps}{2} \right)$ additional vertices. We complete the construction of $G'$ by adding the edge $\{p, q\}$ for all $p \in P$ and $q \in Q$.

    \begin{figure}[t]
        \centering
        \begin{subfigure}[t]{0.48\textwidth}
            \centering
            \includegraphics[width=\textwidth]{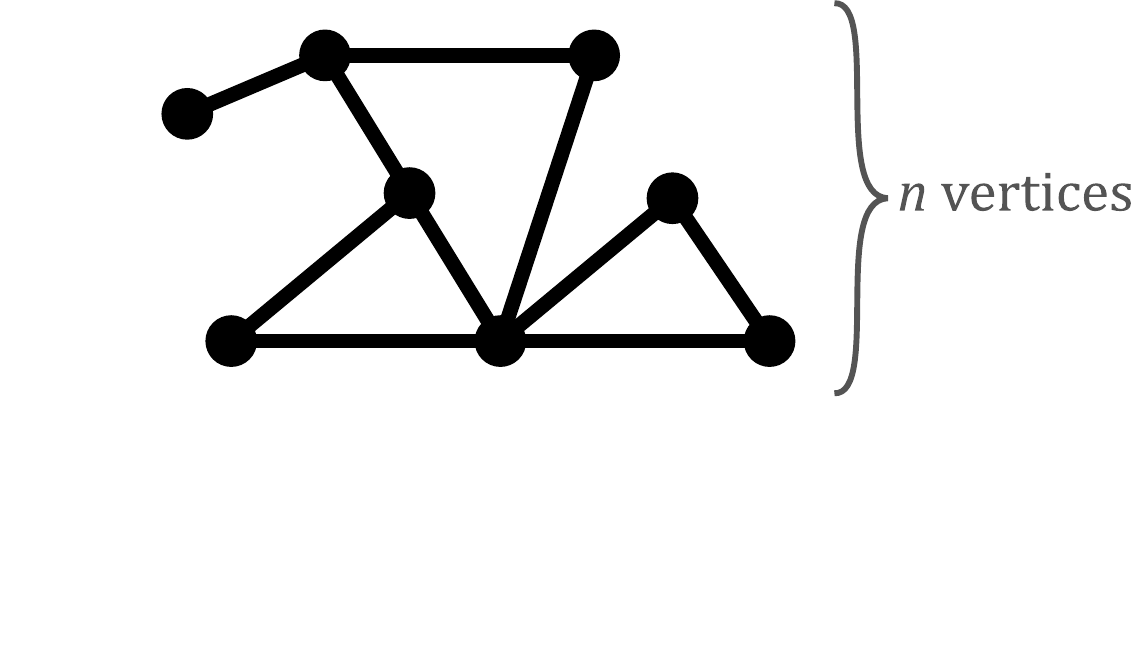}
            \caption{Example input graph $G$ with $n$ vertices.}
        \end{subfigure}
        \hfill
        \begin{subfigure}[t]{0.48\textwidth}
            \centering
            \includegraphics[width=\textwidth]{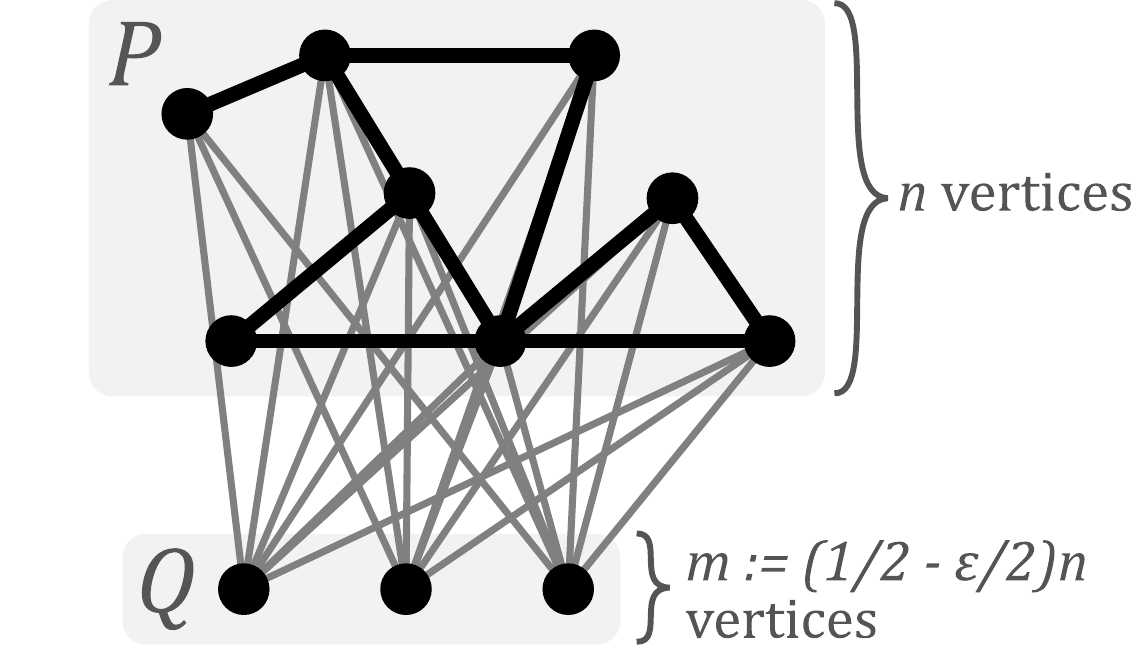}
            \caption{The resulting graph $G'$ when performing the reduction on $G$ with $\eps = \frac14$.}
        \end{subfigure}
        \caption{An example of the reduction for DFVS. }
        \label{fig:vc-reduction}
    \end{figure}

    To prove the correctness of this reduction, we show that the output of an algorithm for $(1.5 - \eps)$-\textsc{Essential detection for VC} on $G'$ can be used as subroutine to distinguish between $\mathrm{VC}(G) \leq \left( \frac12 + \delta \right) n$ and $\mathrm{VC}(G) \geq (1 - \delta) n$ for arbitrarily small~$\delta$. From now on we fix $\delta > 0$ to be arbitrarily small, so that $\delta \leq \frac{\eps}{4}$ in particular.

    Now suppose that an algorithm for $(1.5 - \eps)$-\textsc{Essential detection for VC} exists and let $S \subseteq V(G')$ be its output when run on $G'$ with $k$ set to $n$. We show that the following two claims hold.
    
    \begin{claim} \label{clm:vc-small}
        If $\mathrm{VC}(G) \leq \left( \frac12 + \delta \right) n$, then $|S| < n$.
    \end{claim}
    \begin{claim} \label{clm:vc-large}
        If $\mathrm{VC}(G) \geq (1 - \delta) n$, then $|S| = n$.
    \end{claim}
    
    Then, simply checking the size of the output set $S$ suffices to distinguish between $\mathrm{VC}(G) \leq \left( \frac12 - \frac{\delta}{2} \right) n$ and $\mathrm{VC}(G) \geq (1 - \delta) n$. From \cref{lem:vc-ugc}, we know that this distinction is NP-hard to make under the UGC, meaning that $(1.5 - \eps)$-\textsc{Essential detection for VC} is also NP-hard when assuming the UGC. To prove \cref{thm:vc-detection-hard}, it therefore suffices to prove \cref{clm:vc-small} and \cref{clm:vc-large}.

    \begin{claimproof}[Proof of \cref{clm:vc-small}]
        Suppose that $\mathrm{VC}(G) \leq \left( \frac12 - \delta \right) n$. We show that $G'$ has a vertex cover whose size is strictly smaller than $n$. Then, by property~(G\ref{g1}), we know that $S$ must be a subset of some smallest vertex cover in $G$. This implies that $|S| < n$ and hence proves the claim. Therefore, it suffices and remains to prove that $G'$ has a vertex cover whose size is strictly smaller than~$n$.

        Let $X$ be a smallest vertex cover in $G$. Since $G' - Q = G$, we see that $X$ is a vertex cover in $G' - Q$, which means that $X \cup Q$ is a vertex cover in $G'$. It remains to show that $|X \cup Q| < n$, which is achieved through the following derivation.
        \begin{align*}
            |X \cup Q| &= |X| + |Q| = \mathrm{VC}(G) + (1/2 - \eps / 2) n \\
            &\leq \left( 1/2 + \delta + 1/2 - \eps/2 \right) \\
            &\leq \left( 1/2 + \eps/4 + 1/2 - \eps/2 \right) n < n & \text{since $\delta \leq \eps/4$ \claimqedhere} \\
        \end{align*}
    \end{claimproof}

    \begin{claimproof}[Proof of \cref{clm:vc-large}]
        Suppose that $\mathrm{VC}(G) \geq (1 - \delta) n$. We show that the following two properties hold:
        \begin{enumerate}[(a)]
            \item $P$ is a vertex cover in $G'$. \label{itm:large-vc-a}
            \item All vertices in $P$ are $(1.5 - \eps)$-essential in $G'$. \label{itm:large-vc-b}
        \end{enumerate}
        Before showing that these two properties hold, we argue why doing so would suffice to prove the claim. First note that $P$ is the unique smallest vertex cover in $G'$ if both properties hold: $P$ is a vertex cover by property~(\ref{itm:large-vc-a}) and any solution that does not contain every vertex of $P$ must be at least $(1.5 - \eps)$ times as large by property~(\ref{itm:large-vc-b}). Hence, $\mathrm{VC}(G') = |P| = n$. Then by property~(G\ref{g2}), $S$ (the output of running a $(1.5 - \eps)$-\textsc{Essential detection for VC} algorithm on $G'$ with $k$ set to $n$) must contain all of $P$, so $|S| \geq n$. By property~(G\ref{g1}), no other vertices in can be in $S$, so $|S| = n$.

        We proceed by showing properties~(\ref{itm:large-vc-a}) and~(\ref{itm:large-vc-b}) hold. It is easy to see that $P$ is a vertex cover in $G'$ since $G' - P = G'[Q]$ is an edgeless graph. This shows that property~(\ref{itm:large-vc-a}) holds.

        To show that property~(\ref{itm:large-vc-b}) holds, consider an arbitrary $p \in P$ and let $X'$ be a vertex cover in $G'$ that does not contain $p$. We show that $|X'| > (1.5 - \eps) n$, meaning that $p$ is $(1.5 - \eps)$-essential. As $p \in P$ was chosen arbitrarily, all vertices in $P$ are $(1.5 - \eps)$-essential.

        To prove that $|X'| > (1.5 - \eps)n$, we provide lower bounds for both $|X' \cap P|$ and $|X' \cap Q|$. First we bound $|X' \cap P|$. Since $G'[P] = G$, we must have $|X' \cap P| \geq \mathrm{VC}(G)$, which is at least $(1 - \delta) n$ by assumption.

        Secondly, we prove that $|X' \cap Q| \geq m$. We proceed by contradiction, so suppose otherwise. Then, there must be at least one $q \in Q$ that is not in $X'$. Since $p \notin X'$ by definition of $X'$, $\{p, q\}$ is an edge in $G' - X'$. This contradicts the assumption that $X'$ is a vertex cover in $G'$. Therefore, it follows that $|X' \cap Q| \geq m = \left( \frac12 - \frac{\eps}{2} \right) n$. Now, we can make the following derivation.
        \begin{align*}
            |X'| &= |X' \cap P| + |X' \cap Q| \geq ( 1 - \delta + 1/2 - \eps / 2) n \\
            &\geq (1 - \eps/4 + 1/2 - \eps/2) n > n & \text{since $\delta \leq \eps/4$}
        \end{align*}
        This shows that property~(\ref{itm:large-vc-a}) holds, which, together with property~(\ref{itm:large-vc-a}), proves the claim.
    \end{claimproof}
    This concludes the proof of \cref{thm:vc-detection-hard}.
\end{proof}

This result is quite easily generalized to yield an analogous bound for \textsc{Vertex Multicut}, since the set of solutions for \textsc{Vertex Cover} on a graph~$G$ equals the set of solutions to the \textsc{Vertex Multicut} instance obtained from~$G$ by introducing a terminal pair $(u_1, u_2)$ for every edge $\{u_1, u_2\}$ of~$G$.
\begin{corollary} \label{cor:multicut-lower-bound}
    Assuming the UGC, $(1.5 - \eps)$-\textsc{Essential detection for Vertex Multicut} is NP-hard for any $\eps \in (0, 0.5]$.
\end{corollary}

\section{Conclusion and discussion} \label{sec:conclusion}
We revisited the framework of search-space reduction via the detection of essential vertices. The improved running-time guarantees for fixed-parameter tractable algorithms that result from our detection algorithms give insight into which types of inputs of NP-hard vertex hitting set problems can be solved efficiently and optimally: not only the inputs whose total solution size is small, but also those in which all but a small number of vertices of an optimal solution are essential. Our detection algorithms arise by analyzing the integrality gap for the \textsc{$v$-Avoiding} version of the corresponding LP-relaxation, which only has to be analyzed for inputs in which~$\{v\}$ is a singleton solution. Our results show that the integrality gaps in this setting are much smaller than for the standard linear program of the hitting set formulation.

For \textsc{Directed Feedback Vertex Set}, our lower bound shows that the polynomial-time algorithm that detects $2$-essential vertices is best-possible under the UGC. For \textsc{Vertex Cover} and \textsc{Vertex Multicut}, our lower bounds do not match the existing upper bounds. It would be interesting to close these gaps.

Our positive results rely on standard linear programming formulations of the associated hitting set problems. In several scenarios, algorithms based on the standard linear program can be improved by considering stronger relaxations such as those derived from the Sherali-Adams hierarchy or Lasserre-hierachy (cf.~\cite{Laurent03}). For example, Aprile, Drescher, Fiorini, and Huynh~\cite{AprileDFH23} proved that for the \textsc{Cluster Vertex Deletion} problem (which asks to hit all the induced~$P_3$ subgraphs) the integrality gap of the standard LP-formulation is~$3$, but decreases to~$2.5$ using the first round of the Sherali-Adams hierarchy. Applying~$(1/\eps)^{\Oh(1)}$ rounds further decreases the gap to~$2+\eps$. Can such hierarchies lead to better algorithms for \textsc{$c$-Essential detection}?

So far, the notion of $c$-essentiality has been explored for vertex hitting set problems on graphs. For other optimization problems whose solutions are subsets of objects (for example, edge subsets, or subsets of tasks in a scheduling problem) one can define $c$-essential objects as those contained in all $c$-approximate solutions. Does this notion have interesting applications for problems that are not about graphs?

\bibliography{references}

\begin{thebibliography}{10}

\bibitem{AchterbergBGRW16}
Tobias Achterberg, Robert~E. Bixby, Zonghao Gu, Edward Rothberg, and Dieter Weninger.
\newblock Presolve reductions in mixed integer programming.
\newblock Technical Report 16-44, ZIB, Takustr.7, 14195 Berlin, 2016.
\newblock URL: \url{http://nbn-resolving.de/urn:nbn:de:0297-zib-60370}.

\bibitem{cographs-speed-original}
V.~E. Alekseev.
\newblock On the entropy values of hereditary classes of graphs.
\newblock {\em Discrete Mathematics and Applications}, 3(2):191--200, 1993.
\newblock \href {https://doi.org/doi:10.1515/dma.1993.3.2.191} {\path{doi:doi:10.1515/dma.1993.3.2.191}}.

\bibitem{probabilistic-method}
Noga Alon and Joel~H. Spencer.
\newblock {\em The Probabilistic Method, Third Edition}.
\newblock Wiley-Interscience series in discrete mathematics and optimization. Wiley, 2008.

\bibitem{AprileDFH23}
Manuel Aprile, Matthew Drescher, Samuel Fiorini, and Tony Huynh.
\newblock A tight approximation algorithm for the cluster vertex deletion problem.
\newblock {\em Math. Program.}, 197(2):1069--1091, 2023.
\newblock \href {https://doi.org/10.1007/S10107-021-01744-W} {\path{doi:10.1007/S10107-021-01744-W}}.

\bibitem{BousquetDT18}
Nicolas Bousquet, Jean Daligault, and St{\'{e}}phan Thomass{\'{e}}.
\newblock Multicut is {FPT}.
\newblock {\em {SIAM} J. Comput.}, 47(1):166--207, 2018.
\newblock \href {https://doi.org/10.1137/140961808} {\path{doi:10.1137/140961808}}.

\bibitem{bumpus-essential}
Benjamin~Merlin Bumpus, Bart M.~P. Jansen, and Jari J.~H. de~Kroon.
\newblock Search-space reduction via essential vertices.
\newblock In Shiri Chechik, Gonzalo Navarro, Eva Rotenberg, and Grzegorz Herman, editors, {\em Proceedings of the 30th Annual European Symposium on Algorithms, {ESA} 2022}, volume 244 of {\em LIPIcs}, pages 30:1--30:15. Schloss Dagstuhl - Leibniz-Zentrum f{\"{u}}r Informatik, 2022.
\newblock \href {https://doi.org/10.4230/LIPIcs.ESA.2022.30} {\path{doi:10.4230/LIPIcs.ESA.2022.30}}.

\bibitem{Cai96}
Leizhen Cai.
\newblock Fixed-parameter tractability of graph modification problems for hereditary properties.
\newblock {\em Inf. Process. Lett.}, 58(4):171--176, 1996.
\newblock \href {https://doi.org/10.1016/0020-0190(96)00050-6} {\path{doi:10.1016/0020-0190(96)00050-6}}.

\bibitem{Cai03a}
Leizhen Cai.
\newblock Parameterized complexity of vertex colouring.
\newblock {\em Discret. Appl. Math.}, 127(3):415--429, 2003.
\newblock \href {https://doi.org/10.1016/S0166-218X(02)00242-1} {\path{doi:10.1016/S0166-218X(02)00242-1}}.

\bibitem{ChawlaKKRS06}
Shuchi Chawla, Robert Krauthgamer, Ravi Kumar, Yuval Rabani, and D.~Sivakumar.
\newblock On the hardness of approximating multicut and sparsest-cut.
\newblock {\em Comput. Complex.}, 15(2):94--114, 2006.
\newblock \href {https://doi.org/10.1007/S00037-006-0210-9} {\path{doi:10.1007/S00037-006-0210-9}}.

\bibitem{ChenLLOR08}
Jianer Chen, Yang Liu, Songjian Lu, Barry O'Sullivan, and Igor Razgon.
\newblock A fixed-parameter algorithm for the directed feedback vertex set problem.
\newblock {\em J. {ACM}}, 55(5):21:1--21:19, 2008.
\newblock \href {https://doi.org/10.1145/1411509.1411511} {\path{doi:10.1145/1411509.1411511}}.

\bibitem{CyganFKLMPPS15}
Marek Cygan, Fedor~V. Fomin, Lukasz Kowalik, Daniel Lokshtanov, D{\'{a}}niel Marx, Marcin Pilipczuk, Michal Pilipczuk, and Saket Saurabh.
\newblock {\em Parameterized Algorithms}.
\newblock Springer, 2015.
\newblock \href {https://doi.org/10.1007/978-3-319-21275-3} {\path{doi:10.1007/978-3-319-21275-3}}.

\bibitem{dijkstras-algorithm}
Edsger~W. Dijkstra.
\newblock A note on two problems in connexion with graphs.
\newblock {\em Numerische Mathematik}, 1:269--271, 1959.
\newblock \href {https://doi.org/10.1007/BF01386390} {\path{doi:10.1007/BF01386390}}.

\bibitem{DowneyF12}
Rodney~G. Downey and M.~R. Fellows.
\newblock {\em Parameterized Complexity}.
\newblock Springer Publishing Company, Incorporated, 2012.

\bibitem{Drescher}
Matthew Drescher.
\newblock {\em Two Approaches to Approximation Algorithms for Vertex Deletion Problems}.
\newblock PhD thesis, Universit{\'e} libre de bruxelles, 2021.
\newblock URL: \url{https://knavely.github.io/knavely.gitub.io/thesis.pdf}.

\bibitem{Fellows06}
Michael~R. Fellows.
\newblock The lost continent of polynomial time: {Preprocessing} and kernelization.
\newblock In {\em Proceedings of the 2nd International Workshop on Parameterized and Exact Computation, IWPEC 2006}, pages 276--277, 2006.
\newblock \href {https://doi.org/10.1007/11847250_25} {\path{doi:10.1007/11847250_25}}.

\bibitem{Fernau10}
Henning Fernau.
\newblock A top-down approach to search-trees: Improved algorithmics for 3-hitting set.
\newblock {\em Algorithmica}, 57(1):97--118, 2010.
\newblock \href {https://doi.org/10.1007/S00453-008-9199-6} {\path{doi:10.1007/S00453-008-9199-6}}.

\bibitem{FominLSM19}
Fedor Fomin, Daniel Lokshtanov, Saket Saurabh, and Meirav Zehavi.
\newblock {\em Kernelization: theory of parameterized preprocessing}.
\newblock Cambridge University Press, 2019.

\bibitem{doubling-trick}
Daniel Golovin, Viswanath Nagarajan, and Mohit Singh.
\newblock Approximating the {$k$}-multicut problem.
\newblock In {\em Proceedings of the 17th Annual {ACM-SIAM} Symposium on Discrete Algorithms, {SODA} 2006}, pages 621--630. {ACM} Press, 2006.
\newblock URL: \url{http://dl.acm.org/citation.cfm?id=1109557.1109625}.

\bibitem{GrammGHN04}
Jens Gramm, Jiong Guo, Falk H{\"{u}}ffner, and Rolf Niedermeier.
\newblock Automated generation of search tree algorithms for hard graph modification problems.
\newblock {\em Algorithmica}, 39(4):321--347, 2004.
\newblock \href {https://doi.org/10.1007/S00453-004-1090-5} {\path{doi:10.1007/S00453-004-1090-5}}.

\bibitem{HatzelJLMPSS23}
Meike Hatzel, Lars Jaffke, Paloma~T. Lima, Tom{\'{a}}s Masar{\'{\i}}k, Marcin Pilipczuk, Roohani Sharma, and Manuel Sorge.
\newblock Fixed-parameter tractability of {DIRECTED} {MULTICUT} with three terminal pairs parameterized by the size of the cutset: twin-width meets flow-augmentation.
\newblock In Nikhil Bansal and Viswanath Nagarajan, editors, {\em Proceedings of the 2023 {ACM-SIAM} Symposium on Discrete Algorithms, {SODA} 2023}, pages 3229--3244. {SIAM}, 2023.
\newblock \href {https://doi.org/10.1137/1.9781611977554.CH123} {\path{doi:10.1137/1.9781611977554.CH123}}.

\bibitem{HeggernesHJKV13}
Pinar Heggernes, Pim van~'t Hof, Bart M.~P. Jansen, Stefan Kratsch, and Yngve Villanger.
\newblock Parameterized complexity of vertex deletion into perfect graph classes.
\newblock {\em Theor. Comput. Sci.}, 511:172--180, 2013.
\newblock \href {https://doi.org/10.1016/J.TCS.2012.03.013} {\path{doi:10.1016/J.TCS.2012.03.013}}.

\bibitem{iterative-rounding}
Kamal Jain.
\newblock A factor 2 approximation algorithm for the generalized steiner network problem.
\newblock {\em Comb.}, 21(1):39--60, 2001.
\newblock \href {https://doi.org/10.1007/s004930170004} {\path{doi:10.1007/s004930170004}}.

\bibitem{ugc-intro}
Subhash Khot.
\newblock On the power of unique 2-prover 1-round games.
\newblock In John~H. Reif, editor, {\em Proceedings on 34th Annual {ACM} Symposium on Theory of Computing, STOC 2002}, pages 767--775. {ACM}, 2002.
\newblock \href {https://doi.org/10.1145/509907.510017} {\path{doi:10.1145/509907.510017}}.

\bibitem{vc-ugc}
Subhash Khot and Oded Regev.
\newblock Vertex cover might be hard to approximate to within $2-\varepsilon$.
\newblock {\em J. Comput. Syst. Sci.}, 74(3):335--349, 2008.
\newblock \href {https://doi.org/10.1016/J.JCSS.2007.06.019} {\path{doi:10.1016/J.JCSS.2007.06.019}}.

\bibitem{Laurent03}
Monique Laurent.
\newblock A comparison of the {S}herali-{A}dams, {L}ov{\'{a}}sz-{S}chrijver, and {L}asserre relaxations for 0-1 programming.
\newblock {\em Math. Oper. Res.}, 28(3):470--496, 2003.
\newblock \href {https://doi.org/10.1287/MOOR.28.3.470.16391} {\path{doi:10.1287/MOOR.28.3.470.16391}}.

\bibitem{c-deletion}
John~M. Lewis and Mihalis Yannakakis.
\newblock The node-deletion problem for hereditary properties is {NP}-complete.
\newblock {\em J. Comput. Syst. Sci.}, 20(2):219--230, 1980.
\newblock \href {https://doi.org/10.1016/0022-0000(80)90060-4} {\path{doi:10.1016/0022-0000(80)90060-4}}.

\bibitem{LokshtanovMRSZ21}
Daniel Lokshtanov, Pranabendu Misra, M.~S. Ramanujan, Saket Saurabh, and Meirav Zehavi.
\newblock {FPT}-approximation for {FPT} problems.
\newblock In D{\'{a}}niel Marx, editor, {\em Proceedings of the 2021 {ACM-SIAM} Symposium on Discrete Algorithms, {SODA} 2021}, pages 199--218. {SIAM}, 2021.
\newblock \href {https://doi.org/10.1137/1.9781611976465.14} {\path{doi:10.1137/1.9781611976465.14}}.

\bibitem{cographs-speed-english}
Vadim~V. Lozin, Colin Mayhill, and Victor Zamaraev.
\newblock Locally bounded coverings and factorial properties of graphs.
\newblock {\em Eur. J. Comb.}, 33(4):534--543, 2012.
\newblock \href {https://doi.org/10.1016/J.EJC.2011.10.006} {\path{doi:10.1016/J.EJC.2011.10.006}}.

\bibitem{MarxR14}
D{\'{a}}niel Marx and Igor Razgon.
\newblock Fixed-parameter tractability of multicut parameterized by the size of the cutset.
\newblock {\em {SIAM} J. Comput.}, 43(2):355--388, 2014.
\newblock \href {https://doi.org/10.1137/110855247} {\path{doi:10.1137/110855247}}.

\bibitem{NastosG12}
James Nastos and Yong Gao.
\newblock Bounded search tree algorithms for parametrized cograph deletion: Efficient branching rules by exploiting structures of special graph classes.
\newblock {\em Discret. Math. Algorithms Appl.}, 4(1), 2012.
\newblock \href {https://doi.org/10.1142/S1793830912500085} {\path{doi:10.1142/S1793830912500085}}.

\bibitem{NiedermeierR03}
Rolf Niedermeier and Peter Rossmanith.
\newblock An efficient fixed-parameter algorithm for 3-hitting set.
\newblock {\em J. Discrete Algorithms}, 1(1):89--102, 2003.
\newblock \href {https://doi.org/10.1016/S1570-8667(03)00009-1} {\path{doi:10.1016/S1570-8667(03)00009-1}}.

\bibitem{PilipczukW18a}
Marcin Pilipczuk and Magnus Wahlstr{\"{o}}m.
\newblock Directed multicut is {W[1]}-hard, even for four terminal pairs.
\newblock {\em {ACM} Trans. Comput. Theory}, 10(3):13:1--13:18, 2018.
\newblock \href {https://doi.org/10.1145/3201775} {\path{doi:10.1145/3201775}}.

\bibitem{KarthikLM19}
Karthik~C. S., Bundit Laekhanukit, and Pasin Manurangsi.
\newblock On the parameterized complexity of approximating dominating set.
\newblock {\em J. {ACM}}, 66(5):33:1--33:38, 2019.
\newblock \href {https://doi.org/10.1145/3325116} {\path{doi:10.1145/3325116}}.

\bibitem{schrijver-combopt}
Alexander Schrijver.
\newblock {\em Combinatorial optimization: polyhedra and efficiency}, volume~24.
\newblock Springer, 2003.

\bibitem{ugc-dfvs-gap}
Ola Svensson.
\newblock Hardness of vertex deletion and project scheduling.
\newblock {\em Theory Comput.}, 9:759--781, 2013.
\newblock \href {https://doi.org/10.4086/toc.2013.v009a024} {\path{doi:10.4086/toc.2013.v009a024}}.

\bibitem{vazirani-approx}
Vijay~V. Vazirani.
\newblock {\em Approximation algorithms}.
\newblock Springer, 2001.
\newblock URL: \url{http://www.springer.com/computer/theoretical+computer+science/book/978-3-540-65367-7}.

\bibitem{Weihe98}
Karsten Weihe.
\newblock Covering trains by stations or the power of data reduction.
\newblock In {\em Algorithms and Experiments (ALEX98)}, pages 1--8, 1998.
\newblock URL: \url{https://citeseerx.ist.psu.edu/viewdoc/summary?doi=10.1.1.57.2173}.

\end{thebibliography}

\clearpage 

\appendix

\section{From Essential Detection to Search Space Reduction} \label{sec:essential:to:search:space}
Below we provide the statement by Bumpus et al.~\cite{bumpus-essential} that leverages algorithms for essential detection to provide improved guarantees for the running times of FPT algorithms. The original formulation only considers \textsc{$\mathcal{C}$-Deletion} problems, but the same proof is easily seen to apply to any \textsc{Vertex Hitting Set} problem on graphs.

\begin{theorem}[{\cite[Theorem 5.1]{bumpus-essential}}] \label{thm:non-essentiality}
Let~$\Pi$ be a \textsc{Vertex Hitting Set} problem. 
\begin{itemize}
    \item Let $\mathcal{A}$ be an algorithm that, given a (potentially annotated) graph $G$ and integer $k$, runs in time~$f(k) \cdot |V(G)|^{\Oh(1)}$ for some non-decreasing function~$f$ and returns a minimum-size solution, if there is one of size at most~$k$. 
    \item Let $\mathcal{B}$ be a polynomial-time algorithm for \psdetectFor{c}{$\Pi$}.
\end{itemize}
Then there is an algorithm that, given a (potentially annotated) graph $G$, outputs a minimum-size solution to~$\Pi$ in time $f(\ell) \cdot |V(G)|^{\Oh(1)}$, where~$\ell$ is the number of vertices in an optimal solution that are \emph{not} $c$-essential.
\end{theorem}

\end{document}